\documentclass[runningheads,envcountsame, 10pt]{llncs}

\usepackage[T1]{fontenc}
\usepackage[english]{babel}

\usepackage{xspace}

\usepackage{lmodern}

\usepackage{comment}

\usepackage{authblk}

\usepackage{scalerel}

\usepackage{graphicx}

\usepackage{dcolumn}

\usepackage{booktabs}

\usepackage[shortlabels]{enumitem} %

\usepackage{multicol}

\usepackage{float}

\usepackage{tabularx}

\usepackage{xltabular}

\usepackage{multirow}

\usepackage{amsfonts, amsmath, amsthm, amssymb}

\usepackage{mathtools}

\usepackage{mathrsfs}

\usepackage[all]{xy}

\usepackage{xfrac}
\usepackage{faktor}

\usepackage[nice]{nicefrac}

\usepackage{stmaryrd}

\usepackage{stackrel}

\usepackage{colonequals}

\usepackage[all]{xy}

\usepackage{tikz}
\usetikzlibrary{arrows,arrows.meta,automata,backgrounds,calc,chains,decorations.fractals,decorations.pathreplacing,fadings,fit,folding,mindmap,patterns,plotmarks,positioning,shadows,shapes.geometric,shapes.symbols,through,trees,cd}

\usepackage{extarrows}

\usepackage{bussproofs}

\usepackage{pict2e}

\usepackage{hyperref}
\usepackage[capitalize]{cleveref} 

\usepackage{lineno}
 
\usepackage{csquotes}

\usepackage{breakcites}

\usepackage{makecell}

\newtheoremstyle{plain}{15pt}{15pt}{\itshape}{}{\bfseries}{.}{.5em}{}
\newtheoremstyle{definition}{15pt}{20pt}{}{}{\bfseries}{.}{.5em}{}

\theoremstyle{plain}

\theoremstyle{definition}
\newtheorem{notation}[theorem]{Notation}

\theoremstyle{remark}

\newcommand{\rem}[1]{}  

\makeatletter
\def\orcidID#1{\href{http://orcid.org/#1}{\protect\raisebox{-1.25pt}{\protect\includegraphics{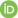}}}}
\makeatother

\newcommand*\circled[1]{
    \tikz[baseline=(char.base), scale=.8, every node/.style={transform shape}]{
        \node[shape=circle,draw,inner sep=2pt] (char) {#1};
    }
}

\newcommand{\customqed}{\hfill $\blacksquare$}

\newcommand{\qmarks}[1]{``#1''}

\DeclareTextFontCommand{\myemph}{\bfseries\em}

\newcolumntype{Y}{>{\centering\arraybackslash}X}

\colorlet{linkcolor}{red!60!black} 
\hypersetup{
    colorlinks=true,
    linkcolor=linkcolor,
    citecolor=linkcolor,
    filecolor=linkcolor,      
    urlcolor=linkcolor,
}

\tikzcdset{scale cd/.style={every label/.append style={scale=#1},
    cells={nodes={scale=#1}}}}

\DeclarePairedDelimiter\set{\{}{\}}
\DeclarePairedDelimiterX\setvbar[2]{\{}{\}}{#1 \nonscript\;\delimsize \vert \nonscript\; #2}
\DeclarePairedDelimiterX\setcolon[2]{\{}{\}}{#1 : #2}

\newcommand{\ov}[1]{\overline{\vphantom{1}#1}}

\newcommand{\eqclass}[2]{\ov{#1}^{\smash{\scriptscriptstyle #2}}} 
\newcommand{\eqS}[1]{\eqclass{#1}{\bbS}}
\newcommand{\eqT}[1]{\eqclass{#1}{\bbT}}

\makeatletter
 \newcommand{\xmapsfrom}[2][]{%
    \ext@arrow3095\leftarrowfill@{#1}{#2}\mapsfromchar
}
\makeatother

\DeclareMathOperator{\id}{id}

\newcommand{\mysetminusD}{\hbox{\tikz{\draw[line width=0.6pt, line cap=round] (3pt, 0) -- (0, 6pt);}}}
\newcommand{\mysetminusT}{\mysetminusD}
\newcommand{\mysetminusS}{\hbox{\tikz{\draw[line width=0.45pt, line cap=round] (2pt, 0) -- (0, 4pt);}}}
\newcommand{\mysetminusSS}{\hbox{\tikz{\draw[line width=0.4pt, line cap=round] (1.5pt, 0) -- (0, 3pt);}}}
\renewcommand{\setminus}{\mathbin{\mathchoice{\mysetminusD}{\mysetminusT}{\mysetminusS}{\mysetminusSS}}}

\renewcommand{\leq}{\leqslant}
\renewcommand{\geq}{\geqslant}

\newcommand{\bb}[1]{\mathbb{#1}}
\newcommand{\N}{\ensuremath{\bb{N}}\xspace}
\newcommand{\Z}{\ensuremath{\bb{Z}}\xspace}

\newcommand{\calR}{\ensuremath{\mathcal{R}}\xspace}
\newcommand{\calRsep}{\ensuremath{\mathcal{R}^{\mathit{sep}}}\xspace}

\let\isom\cong

\newcommand{\defeq}{\vcentcolon=}

\let\lnottemp\lnot
\renewcommand{\lnot}{\lnottemp \hspace*{0.1em}}

\let\oldexists\exists
\let\exists\relax
\newcommand{\exists}{\hspace*{0em}\oldexists\hspace*{0.07em}}
\let\oldforall\forall
\let\forall\relax 
\newcommand{\forall}{\hspace*{0em}\oldforall\hspace*{0.07em}}

  {\gdef\scalefactor{#1}\begin{center}\proofSkipAmount \leavevmode}%
  {\scalebox{\scalefactor}{\DisplayProof}\proofSkipAmount \end{center} }
  
\newcommand{\bbS}{\bb{S}}
\newcommand{\bbT}{\bb{T}}
\newcommand{\bbU}{\bb{U}}
\newcommand{\variables}{\mathcal{V}}
\newcommand{\var}{\mathsf{var}} 
\newcommand{\op}{\mathsf{op}}

\newcommand{\theoryeq}[1]{=_\mathbb{#1}} 

\newcommand{\modST}{modulo~$(\bb{S}, \bb{T})$\xspace}

\newcommand{\Ring}{\ensuremath{\mathsf{Ring}}}
\newcommand{\Monoid}{\ensuremath{\mathsf{Mon}}}
\newcommand{\AbGrp}{\ensuremath{\mathsf{AbGrp}}}

\newcommand{\type}{\mathsf{type}}
\newcommand{\inject}{\mathsf{inject}}
\newcommand{\sep}{\mathsf{sep}}

\newcommand{\cat}[1]{\mathsf{#1}}
\newcommand{\Set}{\cat{Set}}

\newcommand{\cisom}{\isom_{\textsf{conc}}}

\newcommand{\EM}[1]{{\normalfont\textbf{EM}(#1)}} 
\newcommand{\catalg}[1]{{\normalfont\textbf{Alg}(#1)}} 
\mathchardef\mathhyphen="2D                            

\newcommand{\multiset}{{\mathcal{M}}}

\newcommand{\distribution}{{\mathcal{D}}}

\newcommand{\AbGrpMonad}{\mathcal{A}}

\DeclarePairedDelimiter{\brackets}{\llbracket}{\rrbracket}

\newcommand{\term}[2]{\mathcal{T}(#1,#2)}  
\newcommand{\freealgebra}[3]{\nicefrac{\term{#1}{#2}}{#3}}

\newcommand{\freeinter}[2]{\brackets{#2}^{#1}}
\newcommand{\freeinterX}[1]{\freeinter{X}{#1}}

\makeatletter
\newcommand{\adjunction}[4]{%
  #1\colon #2%
  \mathrel{\vcenter{%
    \offinterlineskip\m@th
    \ialign{%
      \hfil$##$\hfil\cr
      \longrightarrow\cr
      \noalign{\kern-.3ex}
      {\scriptscriptstyle\bot}\cr
      \longleftarrow\cr
    }%
  }}%
  #3 \noloc #4%
}

\newcommand\noloc{%
  \nobreak
  \mspace{6mu plus 1mu}
  {:}
  \nonscript\mkern-\thinmuskip
  \mathpunct{}
  \mspace{2mu}
}

\makeatother

\newcommand{\depth}{\mathsf{depth}} 

\newcommand{\rightarrowdbl}{\rightarrow\mathrel{\mkern-14mu}\rightarrow}
\newcommand{\xrightarrowdbl}[2][]{%
  \xrightarrow[#1]{#2}\mathrel{\mkern-14mu}\rightarrow
}
\newcommand{\leftarrowdbl}{\leftarrow\mathrel{\mkern-14mu}\leftarrow}
\newcommand{\xleftarrowdbl}[2][]{%
  \xleftarrow[#1]{#2}\mathrel{\mkern-14mu}\leftarrow
}

\newcommand{\SN}{{\normalfont SN}\xspace}
\newcommand{\CR}{{\normalfont CR}\xspace}
\newcommand{\WCR}{{\normalfont WCR}\xspace}
\newcommand{\CRc}{{\normalfont CR$\circlearrowleft$}\xspace}
\newcommand{\WCRc}{{\normalfont WCR$\circlearrowleft$}\xspace}

\definecolor{cblue}{rgb}{0,0.4,0.7}
\definecolor{clighterblue}{rgb}{0,0.6,1.0}
\colorlet{cred}{red}
\colorlet{cgreen}{green!80!black}
\colorlet{corange}{orange!70!red}
\colorlet{cpureorange}{orange}
\colorlet{cpurple}{clighterblue!50!cred}

\colorlet{clightblue}{clighterblue!50!cblue!40}
\colorlet{clightred}{cred!40}
\colorlet{clightgreen}{cgreen!80!cblue!40}
\colorlet{clightyellow}{corange!40!yellow!50}
\colorlet{clightorange}{cred!50!orange!40}
\colorlet{clightpurple}{clighterblue!50!cred!50}

\colorlet{cdarkred}{cred!70!black}
\colorlet{cdarkgreen}{cgreen!60!black}
\colorlet{cdarkblue}{cblue!60!black}

\begin{document}

\title{Correspondence between Composite Theories and Distributive Laws}
\titlerunning{Correspondence between Composite Theories and Dist.~Laws}

\author{
    Alo\"is Rosset \inst{1} \orcidID{0000-0002-7841-2318} \and
    Maaike Zwart \inst{2} \orcidID{0000-0002-0257-1574} \and \\
    \vspace*{-3mm}
    Helle Hvid Hansen \inst{3} \orcidID{0000-0001-7061-1219} \and
    J\"org Endrullis \inst{1}
    \orcidID{0000-0002-2554-8270}
}
\authorrunning{A.~Rosset, M.~Zwart, H.H.~Hansen, J.~Endrullis}

\institute{
    Vrije Universiteit Amsterdam, Amsterdam, Netherlands \\
    \email{\{a.rosset, j.endrullis\}@vu.nl}
    \and
    IT University of Copenhagen, Copenhagen, Denmark \\
    \email{maaike.annebeth@gmail.com}
    \and
    University of Groningen, Groningen, Netherlands \\
    \email{h.h.hansen@rug.nl}
}

\maketitle

\begin{abstract}
    Composite theories are the algebraic equivalent of distributive laws. In this paper, we delve into the details of this correspondence and concretely show how to construct a composite theory from a distributive law and vice versa.
    Using term rewriting methods, we also describe when a minimal set of equations axiomatises the composite theory.
\end{abstract}

\section{Introduction}

    Monads are categorical structures \cite{Barr_Wells_1985_TTT,MacLane_1971} with many applications in (co)algebraic approaches to program semantics, notably to model effects such as nondeterminism, probabilities and exceptions \cite{Moggi_1991, Plotkin_Power_2001, Bonchi_Sokolova_Vignudelli_2019,Jacobs_Silva_Sokolova_2014}.
    Monads that occur in the specification of programs and are used in reasoning about programs are often finitary and $\Set$-based, and hence can be presented as algebraic theories \cite{Borceux_1994_vol2,Manes_1976,Aczel_Adamek_Milius_Velebil_2003}.
    
    The algebraic view on monads has been especially useful when studying monad compositions \cite{Cheng_2020,Goy_Petrisan_2020,Parlant_2020_Thesis,Pirog_Staton_2017,Zwart_Marsden_journal_2022}.
    Composing monads is a way to combine multiple computational effects, and is usually done categorically via a distributive law \cite{Beck_1969,ManesMulry_2007}. 
    However, the required distributive laws do not always exist, and the use of algebraic theories was instrumental in proving so-called no-go theorems, which tell us when two finitary monads cannot be composed via a distributive law \cite{Zwart_Marsden_journal_2022}. 

    Central to these results is the correspondence between composites of algebraic theories, and distributive laws between the corresponding monads. 
    Briefly stated, a composite of two algebraic theories $\bbS$ and $\bbT$ is a theory $\bbU$ that contains all the function symbols and equations of $\bbS$ and $\bbT$ as well as a set of distribution axioms that specify how equality of mixed terms can be reduced to equality in $\bbS$ and $\bbT$. 
    Composite theories were originally studied by Cheng \cite{Cheng_2020} on the abstract level of Lawvere theories. Pir\'og \& Staton \cite{Pirog_Staton_2017} formulated them in the more concrete setting of algebraic theories. 

    While Pir\'og \& Staton state the correspondence between composite theories and distributive laws, they do not provide a proof, referring instead to Cheng. In her thesis, Zwart \cite{Zwart_2020} gives a constructive version of this correspondence for the category $\Set$, but she does not prove directly that the algebraic theory she constructs from a distributive law is indeed a composite theory.

    Furthermore, the theory Zwart constructs is given via a set $E_\lambda$ that contains all possible equations with interaction between the theories $\bbS$ and $\bbT$.
    While this axiomatisation does the job, it is neither elegant nor practical to work with.
    Composite theories can often be described in terms of a few simple distribution axioms. A classic example is the theory of rings, which is a composite of the theories of monoids and Abelian groups via the two `times over plus' distribution axioms. A systematic approach to identify such a minimal set of distribution axioms for a composite theory would be far more practical than the set $E_\lambda$.

    In this paper, we present a full and self-contained proof of the correspondence between composite theories $\bbU$ (of $\bbT$ after $\bbS$) and distributive laws $\lambda\colon ST \to TS$, where $\bbS$ and $\bbT$ are algebraic theories and $S,T$ are their corresponding finitary $\Set$-monads. \cref{sec:composite_theory_=>_dl} shows how to get a distributive law from a composite theory, and \cref{sec:dl_=>_composite_theory} shows how to construct a composite theory from a distributive law.
    The proof of the latter uses term rewriting techniques.
    In particular, we introduce \emph{functorial rewriting systems} in order to reason about strings of functors, and to obtain a separation of $\bbU$-terms.

    In addition, in \cref{sec:application} we give criteria that ensure that a certain minimal set of distribution axioms $E' \subseteq E_\lambda$ suffices to axiomatise $\bbU$. 
    The natural candidate for $E'$ consists of equations in which the left-hand side is made of exactly one $\bbS$-operation symbol, which is applied to arguments built from up to one $\bbT$-operation symbol.
    We prove that if a term rewriting system corresponding to $E'$ is terminating, then $E_\bbS \cup E_\bbT \cup E'$ axiomatises $\bbU$.
    To illustrate that this criterion is not trivially satisfied, we give an example in which $E'$ does not terminate and indeed does not axiomatise $\bbU$.  
    Finally, we show that we have termination if the right-hand sides of the equations in $E'$ are of a certain form, and apply our results to establish presentations of some composite monads/theories.

\section{Preliminaries}\label{sec:preliminaries}

We assume that the reader is familiar with basic notions of category theory \cite{Awodey_2006,MacLane_1971,Riehl_2017}.
This section recalls basic definitions and results concerning monads, algebraic theories, and term rewriting systems, and fixes notation for the concepts we use in this paper.

\subsection{Monads}

\begin{definition}
    \label{def:monad}
    A \myemph{monad} $(M,\eta,\mu)$ on a category $\cat{C}$ is a triple consisting of an endofunctor $M: \cat{C} \to \cat{C}$, and two natural transformations, the \myemph{unit} $\eta : \id \Rightarrow M$ and the \myemph{multiplication} $\mu : M^2 \Rightarrow M$ that make \eqref{eq:def_monad_axiom_unit} and \eqref{eq:def_monad_axiom_multiplication_associativity} commute.
    For convenience, we often refer to a monad $(M,\eta,\mu)$ by its functor part $M$.

    \noindent\begin{tabularx}{\textwidth}{@{}XX@{}}
        \begin{equation}
            \label{eq:def_monad_axiom_unit}
            \begin{tikzcd}[ampersand replacement=\&, scale cd=.85]
                M \ar[r, "M\eta"] \ar[rd, equal] \& M^2 \ar[d, "\mu"] \& M \ar[l, "\eta M"'] \ar[ld, equal] \\
                \& M \&
            \end{tikzcd}
        \end{equation} &
        \begin{equation}
            \label{eq:def_monad_axiom_multiplication_associativity}
            \begin{tikzcd}[ampersand replacement=\&, scale cd=.85]
                M^3 \ar[d, "M\mu"'] \ar[r, "\mu M"] \& M^2 \ar[d, "\mu"] \\
                M^2 \ar[r, "\mu"'] \& M
            \end{tikzcd}
        \end{equation}
    \end{tabularx}
\end{definition}
\vspace*{-1.5\baselineskip}

\begin{example}
    \label{ex:monads}
    Here are some examples of $\Set$-monads, where we always mean the finitary versions. 
    For more details on these monads, see e.g.~\cite[§1.2.1]{Goy_2021_Thesis}.
    \begin{itemize}[topsep=0pt]
        \item The \emph{list} and \emph{non-empty list} monads $L$ and $L^+$, with $\eta^L_X(x) = \eta^{L^+} (x) = [x]$, and $\smash{\mu^L = \mu^{L^+}}$ being concatenation.
        
        \item The \emph{multiset} monad $\multiset$, with $\eta^\multiset(x) = \Lbag x \Rbag$ and $\mu^\multiset$ taking the union, adding multiplicities.
        Taking multiplicities in $\Z$ gives the \emph{Abelian group} monad $\AbGrpMonad$.
        
        \item The \emph{distribution} monad $\distribution$, with $\eta^\distribution (x) = 1x$ and a weighted average of $\mu^\mathcal{D}$.
        
        \item The \emph{reader} monad $R_A(X) = X^A$, where $A$ is a finite set, with $\eta^R$ the constant function and $\mu^R$ reading the same element twice. 
    \end{itemize}
\end{example}

\begin{definition}
    \label{def:monad_morphism}
    Given two monads $(M,\eta^M,\mu^M)$ and $(T,\eta^T,\mu^T)$ on a category $\cat{C}$, a \myemph{monad morphism} from $M$ to $T$ is a natural transformation $\theta : M \Rightarrow T$ that makes \eqref{eq:def_monad_morphism_axiom_unit} and \eqref{eq:def_monad_morphism_axiom_multiplication} commute, where $\theta \theta \defeq \theta_{T} \cdot M\theta = T\theta \cdot \theta_M$.
    If each component of $\theta$ is an isomorphism, we say that the two monads are \myemph{isomorphic}.

    \vspace*{-.5\baselineskip}
    \noindent\begin{tabularx}{\textwidth}{@{}XX@{}}
        \begin{equation}
            \label{eq:def_monad_morphism_axiom_unit}
            \begin{tikzcd}[row sep = 0.25em, ampersand replacement=\&, scale cd=.9]
                \& M \ar[dd, "\theta"] \\
                \id \ar[ur, "\eta^M"] \ar[dr, "\eta^T"'] \& \\
                \& T
            \end{tikzcd}
        \end{equation} &
        \begin{equation}
            \label{eq:def_monad_morphism_axiom_multiplication}
            \begin{tikzcd}[ampersand replacement=\&,scale cd=.9]
                M^2 \ar[d, "\mu^M"'] \ar[r, "\theta \theta"] \& T^2 \ar[d, "\mu^T"] \\
                M \ar[r, "\theta"'] \& T
            \end{tikzcd}
        \end{equation}
    \end{tabularx}
\end{definition}
\vspace*{-1.5\baselineskip}

\begin{definition}
    \label{def:monad_algebra}
    Let $(M,\eta,\mu)$ be a monad on category $\cat{C}$.
    An (Eilenberg-Moore) $M$-\myemph{algebra} is a $\cat{C}$-morphism $\alpha : MX \to X$ for some $X \in \cat{C}$, denoted $(X,\alpha)$ for short, such that \eqref{eq:def_monad_algebra_axiom_unit} and \eqref{eq:def_monad_algebra_axiom_multiplication_associativity} commute. 
    An $M$-algebra \myemph{homomorphism} $f : (X,\alpha) \to (Y,\beta)$ between two $M$-algebras is a function $f : X \to Y$ such that \eqref{eq:def_monad_algebra_morphism} commutes.
    The category of $M$-algebras and $M$-algebra homomorphisms is denoted $\EM{M}$ and called the \myemph{Eilenberg-Moore category} of $M$.

    \vspace*{-1\baselineskip}
    \noindent\begin{tabularx}{\textwidth}{@{}XXX@{}}
        \begin{equation}
            \label{eq:def_monad_algebra_axiom_unit}
            \begin{tikzcd}[ampersand replacement=\&, scale cd=.9]
                X \ar[r, "\eta_X"] \ar[rd, equal] \& MX \ar[d, "\alpha"] \\
                \& X
            \end{tikzcd}
        \end{equation} 
        &
        \begin{equation}
            \label{eq:def_monad_algebra_axiom_multiplication_associativity}
            \begin{tikzcd}[ampersand replacement=\&, scale cd=.9]
                M^2X \ar[d, "M\alpha"'] \ar[r, "\mu_X"] \& MX \ar[d, "\alpha"] \\
                MX \ar[r, "\alpha"'] \& X
            \end{tikzcd}
        \end{equation}
        &
        \begin{equation}
            \label{eq:def_monad_algebra_morphism}
            \begin{tikzcd}[ampersand replacement=\&,scale cd=.9]
                MX \ar[d, "\alpha"'] \ar[r, "Mf"] \& MY \ar[d, "\beta"] \\
                X \ar[r, "f"'] \& Y
            \end{tikzcd}
        \end{equation}
    \end{tabularx}
\end{definition}
\vspace*{-2\baselineskip}

\begin{definition}
    \label{def:distributive_law}
    Let $S,T$ be monads.
    A \myemph{distributive law} $\lambda : ST \Rightarrow TS$ between monads is a natural transformation satisfying \eqref{eqn:distributive_law_unit_axiom_S}-\eqref{eqn:distributive_law_multiplication_axiom_T}.
    A \myemph{weak distributive law} $\lambda : ST \Rightarrow TS$ is a natural transformation satisfying \eqref{eqn:distributive_law_unit_axiom_T}-\eqref{eqn:distributive_law_multiplication_axiom_T}.

    \vspace*{-0.5\baselineskip}
    \noindent\begin{tabularx}{\textwidth}{@{}XX@{}}
        \begin{equation}
            \label{eqn:distributive_law_unit_axiom_S}
            \begin{tikzcd}[ampersand replacement=\&, scale cd=.9]
                \& T \& \\
                ST \& \& TS
                \ar[from=1-2, to=2-1, "\eta^S T"']
                \ar[from=1-2, to=2-3, "T \eta^S"]
                \ar[from=2-1, to=2-3, "\lambda"]
            \end{tikzcd}
        \end{equation}
        &
        \begin{equation}
            \label{eqn:distributive_law_unit_axiom_T}
            \begin{tikzcd}[ampersand replacement=\&, scale cd=.9]
                \& S \& \\
                ST \& \& TS
                \ar[from=1-2, to=2-1, "S \eta^T"']
                \ar[from=1-2, to=2-3, "\eta^T S"]
                \ar[from=2-1, to=2-3, "\lambda"]
            \end{tikzcd}
        \end{equation}
    \end{tabularx}
    \\[-8mm]
    \noindent\begin{tabularx}{\textwidth}{@{}XX@{}}
        \begin{equation}
            \label{eqn:distributive_law_multiplication_axiom_S}
            \begin{tikzcd}[ampersand replacement=\&, scale cd=.9]
                SST \& STS \& TSS \\
                ST \& \& TS
                \ar[from=1-1, to=1-2, "S \lambda"]
                \ar[from=1-2, to=1-3, "\lambda S"]
                \ar[from=1-1, to=2-1, "\mu^S T"]
                \ar[from=1-3, to=2-3, "T \mu^S"']
                \ar[from=2-1, to=2-3, "\lambda"]
            \end{tikzcd}
        \end{equation}
        &
        \begin{equation}
            \label{eqn:distributive_law_multiplication_axiom_T}
            \begin{tikzcd}[ampersand replacement=\&, scale cd=.9]
                STT \& TST \& TTS \\
                ST \& \& TS
                \ar[from=1-1, to=1-2, "\lambda T"]
                \ar[from=1-2, to=1-3, "T \lambda"]
                \ar[from=1-1, to=2-1, "S \mu^T"]
                \ar[from=1-3, to=2-3, "\mu^T S"']
                \ar[from=2-1, to=2-3, "\lambda"]
            \end{tikzcd}
        \end{equation}
    \end{tabularx}
\end{definition}
\vspace*{-\baselineskip}

A distributive law $\lambda : ST \to TS$ induces a monad structure on the functor $TS$ as follows \cite[§1]{Beck_1969}:
\begin{equation}
    \label{eq:monads_composition_obtained_from_a_distributive_law}
    \Big(TS, 
    \;\; \eta^{TS} \defeq \big(\id{} \xrightarrow{\eta^T \eta^S} TS \big),
    \;\; \mu^{TS} \defeq \big( TSTS \xrightarrow{T \lambda S} TTSS \xrightarrow{\mu^T \mu^S} TS \big) \Big)
\end{equation}

The algebras for this composite monad are algebras that are simultaneously $S$-algebras and $T$-algebras. 
This is visible through the isomorphism $\EM{TS} \isom \catalg{\lambda}$ \cite[§2]{Beck_1969}, where the category $\catalg{\lambda}$ of $\lambda$-algebras is defined as follows:

\bigskip
\noindent\begin{minipage}{.69\linewidth}
\begin{definition}
    \label{def:lambda_algebra}
    Given monads $S, T$ and distributive law $\lambda : ST \to TS$, then the objects of the category $\catalg{\lambda}$ are triples $(X, \sigma, \tau)$, such that $(X, \sigma)$ is an $S$-algebra and $(X, \tau)$ is a $T$-algebra, and the diagram on the right commutes.
    The morphisms of $\catalg{\lambda}$ are $\cat{C}$-morphisms that are both $S$- and $T$-algebra homomorphisms.
\end{definition}
\end{minipage}
\hfill
\begin{minipage}{.29\linewidth}
    \begin{equation*}
        \begin{tikzcd}[ampersand replacement=\&, column sep=tiny, row sep=small]
            STX \& \& TSX \\
            SX \& \& TX \\
            \& X \&
            \ar[from=1-1, to=1-3, "\lambda"]
            \ar[from=1-1, to=2-1, "S\tau"']
            \ar[from=2-1, to=3-2, "\sigma"']
            \ar[from=1-3, to=2-3, "T\sigma"]
            \ar[from=2-3, to=3-2, "\tau"]
        \end{tikzcd}
    \end{equation*}
\end{minipage}

\subsection{Algebraic Theories}

\begin{definition}\label{def:algebraic-theory}
    An \myemph{algebraic theory} is a pair $(\Sigma,E)$ consisting of an algebraic signature $\Sigma$ and set of equations $E$ over $\Sigma$ defined as follows.
    \begin{itemize}[topsep=2pt]
        \item 
        An \myemph{algebraic signature} $\Sigma$ is a set of operation symbols. Each $\op^{(n)} \in \Sigma$ has an arity $n \in \N$.
        
        \item 
        The set $\term{\Sigma}{X}$, also denoted $\Sigma^* X$, of $\Sigma$-\myemph{terms} over a set $X$ is defined inductively: elements in $X$ are terms, and given terms $t_1,\ldots,t_n$ and $\op^{(n)} \in \Sigma$, then $\op(t_1,\ldots,t_n)$ is a term.
        
        \item 
        An \myemph{equation} over a signature $\Sigma$ is a pair $(s,t)$ of $\Sigma$-terms.
    \end{itemize}        
\end{definition}

For the rest of this paper, we fix a set $\variables = \set{v_1, v_2, v_3, \ldots}$ of variables. The subset of $\variables$ appearing in a term $t$ is denoted as $\var(t)$.

\begin{definition}
    The category $\catalg{\Sigma,E}$ consists of $(\Sigma,E)$-algebras and homomorphisms between them.
    \begin{itemize}[topsep=2pt]
        \item 
        A $\Sigma$-\myemph{algebra} is a pair $(X, \brackets{\cdot})$ consisting of a set $X$ and a collection of interpretations: for each $\op^{(n)}  \in \Sigma$, we have $\brackets{\op} : X^n \to X$. Any function $f : X \to Y$ extends to a unique homomorphism,
        $\brackets{\cdot}_f : \term{\Sigma}{X} \to Y$, as given by equations \eqref{eqn:def_[[]]_variables} and \eqref{eqn:def_[[]]_operations} below. When $f = \id_X$, we omit the subscript. Functions of the form $\upsilon : \variables \to Y$ are called \myemph{variable assignments}.
        \begin{align}
            \brackets{x}_f &\defeq f(x), \text{ and} \label{eqn:def_[[]]_variables} \\
            \brackets{\op(t_1,\ldots,t_n)}_f &\defeq \brackets{\op} ( \brackets{t_1}_f, \ldots, \brackets{t_n}_f). \label{eqn:def_[[]]_operations}
        \end{align}
        
        \item 
        A $(\Sigma,E)$-\myemph{algebra} $(X,\brackets{\cdot})$ is a $\Sigma$-algebra whose $\brackets{\cdot}$ satisfies all equations in $E$, i.e., for each $(s,t) \in E$ and all variable assignments $\upsilon$, $\brackets{s}_\upsilon = \brackets{t}_\upsilon$.
        
        \item 
        A $(\Sigma,E)$-algebra \myemph{homomorphism} $f : (X,\brackets{\cdot}) \to (X', \brackets{\cdot}')$ is a function $f : X \to X'$ such that $f \brackets{\op} = \brackets{\op}' f^n$, for all $ \op^{(n)} \in \Sigma$.
    \end{itemize}   
\end{definition}

Given an algebraic theory $\bbT=(\Sigma_\bbT, E_\bbT)$ and $\Sigma_\bbT$-terms $s$ and $t$, we write $s \theoryeq{T} t$ to denote that the equality $s=t$ is derivable from the axioms $E_\bbT$ in equational logic. The inference rules of equational logic are in \cref{tab:equational_logic} (appendix).

\begin{definition}
    We have a free-forgetful adjunction
    $\adjunction{F}{\Set}{\catalg{\Sigma,E}}{U}$ described as follows:
    \begin{itemize}[topsep=2pt]
        \item 
        The \myemph{free $\bbT$-algebra} on a set $X$ is the $(\Sigma_\bbT,E_\bbT)$-algebra $(\freealgebra{\Sigma_\bbT}{X}{\theoryeq{T}},\freeinterX{\cdot})$ with carrier $\term{\Sigma_\bbT}{X}$ modulo $\theoryeq{T}$.
        The equivalence class of a term $t$ is denoted $\eqclass{t}{}$.
        The interpretation of $\op^{(n)} \in \Sigma$ is
        \(
            \freeinterX{\op}(\eqclass{t_1}{}, \ldots, \eqclass{t_n}{}) \defeq \eqclass{\op(t_1,\ldots,t_n)}{}.
        \)
        
        \item 
        The \myemph{free functor} $F: \Set \to \catalg{\Sigma,E}$ sends $X$ to its free $(\Sigma,E)$-algebra.
    \end{itemize}
    Composing the adjoint functors gives a monad $(T \defeq UF, \eta, \mu)$, called the \myemph{free algebra monad}~\cite[VI.1]{MacLane_1971}. 
    The unit is $\eta:x \mapsto \eqclass{x}{}$ and the multiplication is $\mu: \smash{\eqclass{t[\eqclass{t_i}{}/v_i]}{} \mapsto \eqclass{t[t_i/v_i]}{}}$.
    The free algebra monad is defined on functions $f: X \to Y$ as follows, for $x \in X$ and $\op^{(n)} \in \Sigma$:
    \vspace*{-.1\baselineskip}
    \begin{equation}
        \label{eq:free_monad_on_function}
        \begin{aligned}
            Tf(\eqclass{x}{}) &= \eqclass{f(x)}{} \\
            Tf \big( \eqclass{\op(t_1,\ldots,t_n)}{} \big) &= \mu \big( \eqclass{\op(Tf(\eqclass{t_1}{}), \ldots, Tf(\eqclass{t_n}{}))}{} \big)
        \end{aligned}
    \end{equation}
\end{definition}
\vspace*{-.7\baselineskip}

\begin{notation}
    \label{rem:wlog_distinct_variables}
    The following standard (shorthand) notation will be used throughout the paper.
    Given terms $t(x_1,..,x_n)$ and $s_1,\ldots,s_n$, we denote by $t[s_1,...,s_n]$ or by $t[s_i]$ the term $t[h]$ where $h(x_i)=s_i$ for $i=1,..,n$.
    In particular, we will write $t[s_x]$ instead of $t[s_x/x]$, where $x$ ranges over all variables in $t$.
    Moreover, given a family of terms $(t_x[s_{x,y}/y])_{x \in X}$, we will simply write each term $t_x[s_y]$.
    Indeed, we can assume each $t_x$ has distinct variables by choosing the (say $m$) variables of $t_{x_1}$ to be $y_1, \ldots, y_m$, the variables of $t_{x_2}$ to start at $y_{m+1}$, and so on.
\end{notation}
\vspace*{-1\baselineskip}

\begin{definition}[{\cite[Def.~5, Lem.~8]{Rosset_Hansen_Endrullis_2022}}]
    \label{def:algebraic_presentation}
    An algebraic theory $(\Sigma,E)$ is an \myemph{algebraic presentation} of a $\Set$-monad $(M,\eta^M,\mu^M)$ if we have an isomorphism of monads $(T,\eta^T,\mu^T) \isom (M, \eta^M, \mu^M)$, where $T$ is the free algebra monad of $(\Sigma,E)$.
    An equivalent formulation is that both categories of algebras are concretely isomorphic\footnote{\qmarks{concrete} means that both functors of this isomorphism commute with the forgetful functors $\EM{M} \to \Set$ and $\catalg{\Sigma, E} \to \Set$. In other words it sends an $M$-algebra $(X, x : MX \to X)$ to a $(\Sigma,E)$-algebra with same carrier $(X, \brackets{\cdot})$ and vice-versa.}: $\EM{M} \cisom \catalg{\Sigma,E}$. The former isomorphism relates the monads on a syntactic level, whereas the latter relates them semantically.
\end{definition}

Note that a monad can have multiple presentations.
\begin{example}
    \label{ex:algebraic_presentation}
    Here are algebraic presentations of the monads from \cref{ex:monads}.
    \begin{itemize}
        \item The \emph{list} monad $L$ is presented by the theory of \emph{monoids}.
        \item The \emph{non-empty list} monad $L^+$ is presented by the theory of \emph{semigroups}.
        \item The \emph{multiset} monad $\multiset$ is presented by the theory of \emph{commutative monoids}.
        \item The \emph{Abelian group} monad $\AbGrpMonad$ is presented by the theory of \emph{Abelian groups}.
        \item The \emph{distribution} monad $\distribution$ is presented by the theory of \emph{convex algebras}~\cite{Jacobs_2010}.
        \item The \emph{reader} monad $R_A$ is presented by the theory of \emph{local states} \cite{Plotkin_Power_2001} consisting of a single $|A|$-ary operation symbol, satisfying idempotence (e.g. $a * a = a$) and diagonal equations (e.g. $(a * b) * (c * d) = (a * d)$). 
    \end{itemize}

\end{example}

\subsection{Term Rewriting Systems}
\label{subsec:trs}

We only briefly explain the basic concepts and results of term rewriting systems (TRS) that we need in our proofs. For more background, we recommend the book \qmarks{Term Rewriting Systems} by Terese \cite{Terese_2003}.

\begin{definition}
    \label{def:TRS}
    Given a signature $\Sigma$, a  \myemph{rewrite rule} $(l \to r)$ is a pair of $\Sigma$-terms $(l,r)$ such that $l$ 
    is not a variable, and all variables in the right occur also in the left: $\var(l) \supseteq \var(r)$. 
    A \myemph{term rewriting system} $\mathcal{R}=(\Sigma,R)$ consists of a signature $\Sigma$ and a set of rewrite rules $R$.
    The rewrite relation $\to_\mathcal{R}$ is the smallest relation on $\term{\Sigma}{X}$ that contains $\mathcal{R}$ and is closed under substitution and under context\footnote{For the definition of context, see \cite[§2.1.1]{Terese_2003}}.
    We simply write $\to$ when $\mathcal{R}$ is clear from the context. The transitive and reflexive closure of $\to$ is written as $\rightarrowdbl$.
    When all operation symbols in $\Sigma$ have arity $1$, then $\mathcal{R} = (\Sigma,R)$ is called a \myemph{string rewriting system}.
\end{definition}

\begin{example}
    \label{ex:trs}
    Let $\Sigma \defeq \set{0^{(0)}, s^{(1)}, +^{(2)}}$ and $\mathcal{R} = \set{x+0 \to x, \; x+s(y) \to s(x+y)}$.
    A rewrite sequence is for instance
    \[
        s(s(0)) + s(0) \quad\to\quad s(s(s(0)) + 0) \quad\to\quad s(s(s(0))).
    \]
\end{example}

\noindent\begin{minipage}{.8\linewidth}
\begin{definition}
    Let $\mathcal{R} \defeq (\Sigma, R)$ be a TRS.
    \begin{itemize}
        \item 
        $\mathcal{R}$ is \myemph{terminating (SN)} if every rewriting sequence is finite $t_0 \to t_1 \to \ldots \to t_n \not\to$.

        \item 
        $\mathcal{R}$ is \myemph{locally confluent (\WCR)} if for all terms $t_1, t_2, t_3$ with $t_2 \leftarrow t_1 \to t_3$, there exists a term $t_4$ with $t_2 \rightarrowdbl t_4 \leftarrowdbl t_3$.

        \item 
        $\mathcal{R}$ is \myemph{confluent (\CR)} if for all terms $t_1, t_2, t_3$ with ${t_2 \leftarrowdbl t_1 \rightarrowdbl t_3}$, there exists a term $t_4$ with $t_2 \rightarrowdbl t_4 \leftarrowdbl t_3$.
    \end{itemize}
\end{definition}
\end{minipage}
\hfill
\begin{minipage}{.19\linewidth}
    \begin{tabularx}{\textwidth}{@{}Y@{}}
    \begin{tikzcd}[row sep=tiny, column sep=tiny, ampersand replacement=\&, scale cd=.75]
        \& t_1 \& \\
        t_2 \&\& t_3 \\
        \& t_4 \&
        \ar[from=1-2, to=2-1]
        \ar[from=1-2, to=2-3]
        \ar[from=2-1, to=3-2, two heads, dotted]
        \ar[from=2-3, to=3-2, two heads, dotted]
        \ar[from=1-2, to=3-2, draw=none, "(\WCR)" description]
    \end{tikzcd}
    \\
    \vspace*{.01mm}
    \begin{tikzcd}[row sep=tiny, column sep=tiny, ampersand replacement=\&, scale cd=.75]
        \& t_1 \& \\
        t_2 \&\& t_3 \\
        \& t_4 \&
        \ar[from=1-2, to=2-1, two heads]
        \ar[from=1-2, to=2-3, two heads]
        \ar[from=2-1, to=3-2, two heads, dotted]
        \ar[from=2-3, to=3-2, two heads, dotted]
        \ar[from=1-2, to=3-2, draw=none, "(\CR)" description]
    \end{tikzcd}
    \end{tabularx}
\end{minipage}
\medskip

A well-known result says that in the presence of termination, local confluence is enough to entail confluence.

\begin{lemma}[Newman's Lemma]
    \label{lem:newmans_lemma}
    If a TRS is terminating {\normalfont(\SN)} and locally confluent {\normalfont(\WCR)}, then it is also confluent {\normalfont(\CR)}.
\end{lemma}

Two common techniques to prove termination are the \emph{polynomial interpretation} over $\N$~\cite[§6.2.2]{Terese_2003} and the \emph{multiset path order}~\cite{SchneiderKamp_etal_2007_Recursive_path_orders} methods.
The idea of polynomial interpretation over $\N$ is to choose a $\Sigma$-algebra $(\N, \brackets{\cdot})$ where every interpretation $\brackets{\op}$ is a monotone polynomial on $\N$.
If each rule $(l, r)$ of a system is strictly decreasing, $\brackets{l}>\brackets{r}$, then termination follows by well-foundedness of $\N$.

\begin{example}
    The TRS in \cref{ex:trs} is terminating.
    To see this, take as polynomial interpretation for example $\brackets{0} = 1$, $\brackets{s(x)} = x + 1$, and $\brackets{x+y} = x + 2y + 1$.
    These polynomials are monotone and every rule is strictly decreasing:
    \begin{align*}
    \brackets{x+0} = x+1 &> x = \brackets{x}, \\
        \qquad \brackets{x+s(y)} = x+2y+3 &> x+2y+2 = \brackets{s(x+y)}.
    \end{align*}
\end{example}

The multiset path order method uses a decreasing sequence of multisets to show termination. We explain this briefly in the appendix.

A common technique for proving local confluence is to prove convergence of \emph{critical pairs}~\cite[§2.7]{Terese_2003}.
Informally, a critical pair is formed when two rewrite rules can be applied to the same term while overlapping on a function symbol, creating two different outcomes.
A critical pair  \emph{converges} if the two resulting terms can be rewritten to a common term.

\begin{lemma}[Critical pair lemma]
    \label{lem:critical_pair_lemma}
    A TRS is locally confluent {\normalfont(\WCR)} if and only if all its critical pairs converge.
\end{lemma}

\section{Composite Theories}
\label{sec:composite_theories}

We introduce the concept of \emph{composite theories}. Our definition is slightly different from, but equivalent to, the original definition by Pir\'og \& Staton~{\cite[Def.~3]{Pirog_Staton_2017}} and equivalent formulations in Zwart's thesis \cite[Def.~3.2, Prop.~3.4]{Zwart_2020}.
\vspace*{-.2\baselineskip}

\begin{definition}\label{def:composite_theory}
    Let $\bbU, \bbS, \bbT$ be algebraic theories.
    Suppose $\bbU$ contains $\bbS$ and $\bbT$, meaning $\Sigma_\bbS, \Sigma_\bbT \subseteq \Sigma_\bbU$ and $E_\bbS, E_\bbT \subseteq E_\bbU$.
    \begin{itemize}[itemsep=2pt, topsep=2pt]
        \item 
        A $\bbU$-term is \myemph{separated} if it is of the form $t[s_x /x]$, where $t$ is a $\bbT$-term and $\setvbar{s_x}{x \in \var(t)}$  is a family of $\bbS$-terms.
        
        \item
        Two separated terms $t[s_x]$ and $t'[s'_y]$ are \myemph{equal \modST} if their $TS$-equivalence classes are equal in $TS\variables$: $\smash{\eqT{t[\eqS{s_x}]} = \eqT{t'[\eqS{s'_y}]}}$.
        
        \item
        $\bbU$ is a \myemph{composite theory} of $\bbT$ after $\bbS$ if every $\bbU$-term $u$ is equal to a separated term $u \theoryeq{U} t[s_x / x]$, that we call a \myemph{separation} of $u$,
        and for any two separated terms $v,v'$, if $v \theoryeq{U} v'$ then $v$ and $v'$ must be equal \modST.
    \end{itemize}
\end{definition}

\begin{proposition}
    \label{prop:essuniqformulations}
    \label{lem:same_equivalence_class_TSX_implies_a_to_d}
    For any two separated terms $t[s_x /x]$ and $t'[s_y /y]$ in a composite theory, the following are equivalent:
    \begin{enumerate}[itemsep=1pt]
        \item \label{it:eqclass-es}
        $t[s_x /x]$ and $t'[s'_y /y]$ are equal \modST in the sense of \cref{def:composite_theory}.

        \item \label{it:dan-es} 
        $t[s_x /x]$ and $t'[s'_y /y]$ are equal \modST in the sense of \cite[Def.~3.2]{Zwart_2020}.
        \customqed\footnote{
            The symbol $\blacksquare$ denotes that the proof is in the Appendix.   
        }
    \end{enumerate}
\end{proposition}

\begin{example}
    Two $\bbS$-terms $s$ and $s'$ are equal \modST if and only if $s \theoryeq{S} s'$, and similarly for $\bbT$-terms.
\end{example}

\begin{example} \label{ex:monoid_abelian_ring}
    The prime example of a composite theory is the theory of rings $\bbU \defeq \mathsf{Ring}$.
    It contains the theories $\bbS \defeq \mathsf{Mon}$ of monoids and $\bbT \defeq \mathsf{AbGrp}$ of Abelian groups.
    We recall their signatures to fix notation: $\Sigma_\mathsf{Mon} \defeq \set{\cdot^{(2)}, 1^{(0)}}$ and $\Sigma_\mathsf{AbGrp} \defeq \set{0^{(0)}, +^{(2)}, -^{(1)}}$.
    We sometimes omit the \qmarks{multiplication} symbol $\cdot$ for simplicity.
    The signature of rings is given by $\Sigma_\mathsf{Ring} \defeq \Sigma_\mathsf{Mon} \uplus \Sigma_\mathsf{AbGrp}$, and the equations are given by the equations of monoids, Abelian groups, and two distributivity axioms:
    \[
        E_\mathsf{Ring} \defeq E_\mathsf{Mon} \cup E_\mathsf{AbGrp} \cup
        \left\{
        \begin{aligned}
            x (y + z) &= (xy) + (xz), \\
            (y + z)x &= (yx) + (zx)
        \end{aligned}
        \right\}.
    \]
    
    A separated term $t[s_x/x]$ in $\mathsf{Ring}$ is an Abelian group term $t$, with  monoid terms $\set{s_x}$ substituted for its variables.
    We give some examples of non-separated terms, of possible separations for them, and of equality 
    modulo $(\mathsf{Mon}, \mathsf{AbGrp})$ between the separations.
    The term $x (y + z)$ is non-separated.
    Possible separations are e.g.~$xy + xz$ and $(xy + xz) + 0$.
    Both are equal modulo $(\mathsf{Mon}, \mathsf{AbGrp})$, as their monoid parts are identical and their Abelian group parts $t = (x_1 + x_2) + 0$ and $t' = x_1 + x_2$ are equal in the theory of Abelian groups.

    The term $x \cdot 0$ is also non-separated.
    It is equal in $\Ring$ to the separated terms $0$ and $(1\cdot x) + (-(x \cdot 1))$. 
    To see that these separations are equal modulo $(\mathsf{Mon}, \mathsf{AbGrp})$, notice that $1 \cdot x \theoryeq{\mathsf{Mon}} x \cdot 1$, and that the terms $0$ and $x_1 + (-x_2)$ are equal in Abelian groups when $x_1 = x_2$. Hence: $\eqclass{0}{\mathsf{AbGrp}} = \eqclass{(\eqclass{1\cdot x}{\mathsf{Mon}}) + (-(\eqclass{x \cdot 1}{\mathsf{Mon}}))}{\mathsf{AbGrp}}$.
\end{example}

We now show that distributive laws between monads correspond one-to-one to composite theories.

\section{From Composite Theory to Distributive Law}
\label{sec:composite_theory_=>_dl}

We first show how to construct a distributive law from a given composite theory.

\begin{theorem}\cite[Theorem 3.8]{Zwart_2020}
    \label{thm:composite_=>_dl}
    Let $\bbS, \bbT$ be algebraic theories with free algebra monads $S,T$ respectively.
    Let $\bbU$ be a composite theory of $\bbT$ after $\bbS$, with free algebra monad $U$.
    Then the following defines a distributive law $\lambda : ST \Rightarrow TS$ such that $\bbU$ is an algebraic presentation of the resulting monad $TS$, where $t'[s'_x]$ is a separation of $s[t_x]$:
    \[
        \lambda_\variables : ST\variables \to TS\variables : \quad
        \eqclass{s[\eqclass{t_x}{\bbT}/x]}{\bbS} \mapsto \eqclass{t'[\eqclass{s'_x}{\bbS}/x]}{\bbT}
    \]
\end{theorem}

\begin{proof}
    Instead of directly checking the axioms for a distributive law, we prove an equivalent characterisation given by Beck~\cite[p.122]{Beck_1969}.
    That is, we claim that there exist a natural transformation $\mu^{TS} : TSTS \Rightarrow TS$ such that: 
    
    \begin{enumerate}[label=(\roman*)]
        \item 
        $(TS, \eta^{TS} \defeq \eta^T\eta^S, \mu^{TS})$ is a monad.
        
        \item 
        The natural transformations $\eta^T S$ and $T\eta^S$ are monad morphisms.
        
        \item 
        The middle unitary law holds: $\mu^{TS}\cdot T\eta^S\eta^TS = Id_{TS}$.
    \end{enumerate}

    It follows then that the monad $(TS, \eta^T\eta^S, \mu^{TS})$ does indeed come from a distributive law, which is given by:
    \(
        \lambda = \mu^{TS} \cdot \eta^T ST \eta^S
    \)
    A simple but tedious calculation shows that indeed $\lambda (\eqclass{s[\eqclass{t_x}{\bbT}/x]}{\bbS} ) = \eqclass{t'[\eqclass{s'_x}{\bbS}/x]}{\bbT}$.
    The details of this calculation are in the appendix.

    To define $\mu^{TS}$, we use the fact that the \textit{functors} $U$ and $TS$ are isomorphic.
    Indeed, since $\bbU$ is a composite theory, every $\bbU$-term $u$ 
    has a separation $u \theoryeq{U} t[s_x/x]$.
    Hence $\phi : U \Rightarrow TS$ and $\psi : TS \Rightarrow U$ given below are inverse natural transformations. 
    Using $\phi, \psi$, and the multiplication $\mu^U$, we can then define $\mu^{TS}$. 
    \begin{align}
        \phi(u) &\defeq 
        \eqclass{t[\eqclass{s_x}{\bbS}/x]}{\bbT}
        \label{eq:def_phi} \\
        \psi(\eqclass{t[\eqclass{s_x}{\bbS}/x]}{\bbT}) &\defeq 
        \eqclass{t[s_x/x]}{\bbU} 
        \label{eq:def_psi}\\
        \mu^{TS} &\defeq \label{eq:def_muTS}
        \big( 
            TSTS \xrightarrow{\smash{\psi\psi}} 
            UU \xrightarrow{\smash{\mu^U}} 
            U \xrightarrow{\smash{\phi}} 
            TS 
        \big).
    \end{align}
    The proofs of $(i)$-$(iii)$ are in the appendix.
\end{proof}

\section{From Distributive Law to Composite Theory}
\label{sec:dl_=>_composite_theory}

We now show how to construct a composite theory from a given distributive law.

\begin{theorem}
    \label{def:U_lambda}
    \label{thm:dl_=>_composite_theory}
    Let $S,T$ be two monads algebraically presented by two algebraic theories $\bbS$ and $\bbT$, respectively.
    Let $\lambda : ST \Rightarrow TS$ be a distributive law.
    We define a set $E_\lambda$ of equations and a theory $\bbU^\lambda$ as follows~\cite[Definition 3.8]{Zwart_2020}.
    \begin{align*}
        E_\lambda \defeq 
        \Big\{
            \big(s[t_x/x],t[s_y/y]\big) 
            & \; \vert \; \lambda_\variables\big( \eqclass{s[\eqclass{t_x}\bbT/x]}\bbS \big) = \eqclass{t[\eqclass{s_y}\bbS/y]}\bbT \big) 
        \Big\}. \\
        \Sigma_{\bbU^\lambda} &\defeq \Sigma_\bbS \uplus \Sigma_\bbT, \\
        E_{\bbU^\lambda} &\defeq E_\bbS \cup E_\bbT \cup E_\lambda.
    \end{align*}
    Then, $\bbU^\lambda$ is a composite theory of $\bbT$ after $\bbS$.
\end{theorem}

To prove \cref{thm:dl_=>_composite_theory}, we observe that every $\bbU^\lambda$-term $u$ can be assigned a regular expression $\type(u)$ of the form $(S \cup T)^*\variables$ expressing how $u$ nests $\bbS$ and $\bbT$ operation symbols. 
We give an example below in \cref{ex:type_inject}.
We obtain a $TS$-separated term by first mapping $u$ to the equivalence class $\eqclass{u}{}$ in $\type(u)$, now viewed as a set.
We then apply $\lambda, \mu^S$ and $\mu^T$ to $\eqclass{u}{}$ until we reach an equivalence class $\xi \in TS\variables$. 
The axioms of the three natural transformations ensure that $\xi$ does not depend on the order in which they were applied.
Finally, by the axiom of choice we can choose a representative of $\xi$ which serves as a separation of $u$.

The termination of the procedure of applying $\lambda, \mu^S$ and $\mu^T$ and the uniqueness of $\xi$ are intuitively clear, yet showing it formally is not trivial. 
In the following definitions we formalise the separation procedure that we described here. We then give a proof of termination using rewriting techniques.
We denote string concatenation with \qmarks{$::$}.

\begin{definition}
    We define a function $\type : \Sigma^*_{\bbU^\lambda}\variables \to \set{S,T}^*\variables$ recursively:
    \begin{itemize}
        \item 
        For $v \in \variables$, then $\type(v) \defeq \variables$.

        \item 
        For $s[u_1, \ldots, u_n]$, where $s \in \term{\Sigma_\bbS}{\variables}$, and $u_1, \ldots, u_n \in \Sigma^*_{\bbU^\lambda}\variables$ such that their root is not an $\bbS$-symbol, let $w$ be longest word in the set $\set{\type(u_1), \ldots, \type(u_n)}$, then ${\type(s[u_1,\ldots, u_n]) \defeq S :: w}$.

        \item 
        The $t[u_1, \ldots, u_n]$ case, where $u_1, \ldots, u_n$ do not start with a $\bbT$-symbol, is dual.
    \end{itemize}
\end{definition}

\begin{definition}
    We recursively define a dependent partial function 
    \[
        \textstyle \inject: \smash{\prod_{(u,w) \in \Sigma^*_{\bbU^\lambda}\variables \times \set{S,T}^*\variables}} \rightharpoonup w.
    \]
    \begin{itemize}
        \item 
        For $v \in \variables$, $\inject(v,\variables) \defeq v$, and 
        \begin{align*}
         \inject(v,S::w') & \defeq \eqS{\inject(v,w')} &  \inject(v,T::w')  & \defeq \eqT{\inject(v,w')}
        \end{align*}
        
        \item 
        For $s[u_1, \ldots, u_n]$ where $s \in \Sigma_\bbS^* \variables$, and $u_1, \ldots, u_n \in \Sigma^*_{\bbU^\lambda}\variables$ being either variables or having root symbols in $\Sigma_\bbT$, 
        \[
            \inject(s[u_1, \ldots, u_n], S::w') \defeq \eqclass{s[\inject(u_1,w'), \ldots, \inject(u_n,w')]}\bbS.
        \]
        
        \item 
        The $t[u_1, \ldots, u_n]$ case, where $u_1, \ldots, u_n$ do not start with a $\bbT$-symbol, is dual.
    \end{itemize}
\end{definition}

\noindent\begin{minipage}{.8\linewidth}
\begin{example}
    \label{ex:type_inject}
    Take the operations $s^{(2)}, s'^{(1)} \in \Sigma_\bbS$, and $t^{(1)} \in \Sigma_\bbT$. For $u \defeq s(s(x,t(x)),t(s'(s(x,x))))$, we have
    \begin{align*}
        w &\defeq \type(u) = STS\variables. \\
        \zeta &\defeq \inject(u,w) = \eqS{s(s(\,\eqT{\eqS{x}}, \eqT{t(\eqS{x})}), \eqT{t(\eqS{s'(s(x,x))})} )}.
    \end{align*}
\end{example}
\end{minipage}
\hfill
\begin{minipage}{.19\linewidth}
    \hfill%
\definecolor{cblue}{rgb}{0,0.4,0.7}
\definecolor{clighterblue}{rgb}{0,0.6,1.0}
\colorlet{cred}{red}
\colorlet{cgreen}{green!80!black}
\colorlet{corange}{orange!70!red}
\colorlet{cpureorange}{orange}
\colorlet{cpurple}{clighterblue!50!cred}

\colorlet{clightblue}{clighterblue!50!cblue!40}
\colorlet{clightred}{cred!40}
\colorlet{clightgreen}{cgreen!80!cblue!40}
\colorlet{clightyellow}{corange!40!yellow!50}
\colorlet{clightorange}{cred!50!orange!40}
\colorlet{clightpurple}{clighterblue!50!cred!50}

\colorlet{cdarkred}{cred!70!black}
\colorlet{cdarkgreen}{cgreen!60!black}
\colorlet{cdarkblue}{cblue!60!black}

\begin{tikzpicture}[
  baseline=-10mm,-,
  level distance=6mm,
  level 1/.style={sibling distance=15mm},
  level 2/.style={sibling distance=8mm},
  s/.style={fill=cblue!20},
  t/.style={fill=cgreen!20},
  scale=.8,
  every node/.style={transform shape}
  ]
  \node (s) {$s$}
    child {node (s') {$s$}
      child {node (x0) {$x$}}
      child {node (t0) {$t$}
        child {node (x1) {$x$}}
      }
    }
    child {node (t) {$t$}
      child {node (s'') {$s'$}
        child {node (s''') {$s$}
          child {node (x2) {$x$}}
          child {node (x3) {$x$}}
        }
      }
    };
  
  \begin{scope}[rounded corners=2mm,on background layer]
    \draw [s] (s.north) to (s.north west) to (s'.north west) to (s'.south west) to (s'.south east) to (s.south east) to (s.north east) to (s.north);
    \draw [s] (s''.north) to (s''.north west) to (x2.south west) to (x3.south east) to (s''.north east) to (s''.north);
    \draw [t] (t.north) to (t.north west) to (t.south west) to (t.south east) to (t.north east) to (t.north);
    \draw [t] (t0.north) to (t0.north west) to (t0.south west) to (t0.south east) to (t0.north east) to (t0.north);
    \draw [s] (x1.north) to (x1.north west) to (x1.south west) to (x1.south east) to (x1.north east) to (x1.north);
    \draw [t] (x0) circle (3mm);
    \draw [s] (x0.north) to (x0.north west) to (x0.south west) to (x0.south east) to (x0.north east) to (x0.north);
  \end{scope}
\end{tikzpicture}
\end{minipage}
\smallskip

Before we formalise the remainder of the separation procedure, we interpret functors and natural transformations as a term rewriting system.

\begin{definition}
    \label{def:functorial_rewriting_system}
    Let $\Sigma \defeq \setvbar{F_i}{i \in I}$ be a finite set of functors, and $\calR \defeq \setvbar{\alpha_j : L_{j} \to R_{j}}{L_j, R_j \in \Sigma^*, j \in J}$ be a finite set of natural transformations.
    We call $(\Sigma, \calR)$ a \myemph{functorial rewriting system}.
\end{definition}

The name \qmarks{functorial rewriting system} is motivated by seeing each natural transformation $\alpha : L \to R$ as a rewrite rule on strings of functors in $\Sigma^*$ as follows. For any strings of functors $A,B \in \Sigma^*$, we can make the rewrite step 
\begin{equation}
    \label{eq:functorial_rewrite_step}
    ALB \xrightarrow{A \alpha B} ARB,
\end{equation} 
where $A$ is seen as a left-context and $B$ as a right-context.
Note that for a functorial rewriting system $(\Sigma, \calR)$ the only valid rewrite steps are those resulting from natural transformations in $\calR$. If the functors in $\Sigma$ satisfy identities like $FG = H$ that are not represented by some $\beta \in \calR$, then we do not allow rewrite steps that use this identity.

We use the following functorial rewriting system for our separation procedure.

\begin{definition}
    \label{def:functorial_rewriting_system_lambda_muS_muT}
    We define a functorial rewriting system $\calRsep = (\Sigma, R)$, where $\Sigma \defeq \set{S,T}$ and
    \(
        R \defeq \set{
            \circled{1} \, \lambda: ST \to TS, \;
            \circled{2} \, \mu^S: SS \to S, \; 
            \circled{3} \, \mu^T: TT \to T
        }.
    \)
\end{definition}

The functors and natural transformations in a functorial rewriting system carry categorical structure in the form of commuting diagrams. We use this information 
to define a variation
of (local) confluence.

\begin{definition}
    A functorial rewriting system is (read $\circlearrowleft$ as \qmarks{commuting})
    \begin{itemize}[topsep=0pt]
        \item
        \WCRc if for all $T_0 \xleftarrow{\alpha} \cdot \xrightarrow{\beta} T_1$ there exists $T_0 \xrightarrowdbl{\gamma} \cdot \xleftarrowdbl{\delta} T_1$ s.t.~$\gamma\alpha = \delta \beta$.
        
        \item
        \CRc if for all $T_0 \xleftarrowdbl{\mathclap{\alpha}} \cdot \xrightarrowdbl{\beta} T_1$ there exists $T_0 \xrightarrowdbl{\gamma} \cdot \xleftarrowdbl{\delta} T_1$ s.t.~$\gamma\alpha = \delta \beta$.
    \end{itemize}
\end{definition}

\begin{remark}
    When writing $T_0 \rightarrowdbl \cdot \leftarrowdbl T_1$, the common descendent must be the exact same word in $\Sigma^*$.
    No equality property between functors is allowed.
\end{remark}

There are equivalents to Newman's Lemma (\cref{lem:newmans_lemma}) and the Critical Pair Lemma (\cref{lem:critical_pair_lemma}).
The proofs are in the appendix.

\begin{lemma}[Functorial Newman's lemma] 
    \label{lem:functorial_newmans_lemma}
    If a functorial rewriting system is terminating {\normalfont(\SN)} and \WCRc, then it is \CRc.
    \customqed
\end{lemma}

\begin{lemma}[Functorial critical pair lemma]
    \label{lem:functorial_critical_pair_lemma}
    A functorial rewriting system is \WCRc if and only if all critical pairs converge with a commuting diagram.
    \customqed
\end{lemma}

The functorial rewriting system $\calRsep$ of \cref{def:functorial_rewriting_system_lambda_muS_muT} has nice properties.

\begin{lemma}
        $\calRsep$ is terminating {\normalfont(\SN)} and confluent-commuting {\normalfont(\CRc)}.
    
\end{lemma}

\begin{proof}
    We show termination (\SN) of $\calRsep$ using polynomial interpretation over $\N$.
    Let $\brackets{S}(x) \defeq 2x+1$ and $\brackets{T}(x) \defeq x+1$, which are indeed monotone in $x$.
    The three rewrite rules are strictly decreasing with respect to that order:
    \begin{align*}
        \brackets{ST}(x) = 2x + 3 &> 2x + 2 = \brackets{TS}(x), \\
        \brackets{SS}(x) = 4x + 3 &> 2x + 1 = \brackets{S}(x), \\
        \brackets{TT}(x) = x + 2 &> x + 1 = \brackets{T}(x).
    \end{align*}
    We now want to prove that $\calRsep$ is \CRc.
    Since we have termination (\SN) it suffices to prove \WCRc by \cref{lem:functorial_newmans_lemma}. 
    We use \cref{lem:functorial_critical_pair_lemma} and notice that there are only $4$ critical pairs, and all converge with a commuting diagram:
    \[
            \begin{tikzcd}[ampersand replacement=\&, sep=small]
                \& SST \\
                STS \&\& ST \\
                TSS \\
                \& TS
                \arrow["{S \lambda ~ \circled{1}}"', from=1-2, to=2-1]
                \arrow["{\circled{2} ~ \mu^S T}", from=1-2, to=2-3]
                \arrow["{\lambda}", from=2-3, to=4-2]
                \arrow["{\lambda S}"', from=2-1, to=3-1]
                \arrow["{T \mu^S}"', from=3-1, to=4-2]
                \arrow[phantom, "\scriptstyle (\lambda \text{ axiom})", from=1-2, to=4-2]
            \end{tikzcd}
            \qquad
            \begin{tikzcd}[ampersand replacement=\&, sep=small]
                \& STT \\
                ST \&\& TST \\
                \&\& TTS \\
                \& TS
                \arrow["{S \mu^T ~ \circled{3}}"', from=1-2, to=2-1]
                \arrow["{\lambda}"', from=2-1, to=4-2]
                \arrow["{\circled{1} ~ \lambda T}", from=1-2, to=2-3]
                \arrow["{T \lambda}", from=2-3, to=3-3]
                \arrow["{\mu^T S}", from=3-3, to=4-2]
                \arrow[phantom, "\scriptstyle (\lambda \text{ axiom})", from=1-2, to=4-2]
            \end{tikzcd}
    \]
    \[
            \begin{tikzcd}[ampersand replacement=\&, sep=small]
                \& SSS \\
                SS \&\& SS \\
                \& S
                \arrow["{S \mu^S ~ \circled{2}}"', from=1-2, to=2-1]
                \arrow["{\circled{2} ~ \mu^S S}", from=1-2, to=2-3]
                \arrow["{\mu^S}", from=2-3, to=3-2]
                \arrow["{\mu^S}"', from=2-1, to=3-2]
                \arrow[phantom, "\scriptstyle (S \text{ axiom})", from=1-2, to=3-2]
            \end{tikzcd}
            \qquad
            \begin{tikzcd}[ampersand replacement=\&, sep=small]
                \& TTT \\
                TT \&\& TT \\
                \& T
                \arrow["{T \mu^T ~ \circled{3}}"', from=1-2, to=2-1]
                \arrow["{\circled{3} ~ \mu^T T}", from=1-2, to=2-3]
                \arrow["{\mu^T}", from=2-3, to=3-2]
                \arrow["{\mu^T}"', from=2-1, to=3-2]
                \arrow[phantom, "\scriptstyle (T \text{ axiom})", from=1-2, to=3-2]
            \end{tikzcd} 
        \qedhere
    \]
\end{proof}

We now have the required tools
to formalise the separation procedure and show that every term in $\bbU^\lambda$ can be separated.

\begin{lemma}
    \label{lem:dl_=>_composite_theory_lem1}
    There is a function $\sep\colon \Sigma^*_{\bbU^\lambda} \to \Sigma^*_{\bbU^\lambda}$ that maps a $\bbU^\lambda$-term $u$ to a $TS$-separated term $\sep(u)$ that is $\bbU^\lambda$-equal to $u$. 
\end{lemma}

\begin{proof}
    Let $u$ be an arbitrary $\bbU^\lambda$-term, and
    let  $w \in (S \cup T)^*$ be such that $\type(u) = w\variables$. 
    Furthermore, let $\zeta \defeq \inject(u, \type(u)) \in \type(u)$.
    Since $\calRsep$ is terminating (\SN), the string $w$ has a normal form $w'$
    which must be $TS$, $T$, or $S$, as any other type will contain a reducible expression (redex).

    Let $\alpha$ be a $\calRsep$-rewrite sequence $w \rightarrowdbl w'$.
    Recall that $\alpha$ is a natural transformation composed of $\lambda$, $\mu^S$ and $\mu^T$.
    Let $\xi \defeq \alpha_\variables (\zeta)$.
    Since $\xi$ is an equivalence class in $TS\variables$, $T\variables$ or $S\variables$, all of its representatives must be separated.
    By the axiom of choice we have functions $\rho_{TS\variables}\colon TS\variables \to \Sigma_T^*\Sigma_S^*\variables$, 
    $\rho_{S\variables}\colon S\variables \to \Sigma_S^*\variables$, and
    $\rho_{T\variables}\colon T\variables \to \Sigma_T^*\variables$, 
    that choose a representative.
    We define $\sep(u)$ to be
    $\rho_{w'\variables} (\xi))$, that is, 
    $\sep(u) \defeq \rho_{w'\variables} ( \alpha_\variables ( \inject(u,\type(u))))$.
    
    To see that $\sep(u)$ is well defined, let $w''$ be another normal form  and $\beta$ be a rewrite sequence $w \rightarrowdbl w''$.
    By \CRc, we necessarily have $w' = w''$, and also $\alpha = \beta$.
    \[
        \begin{tikzcd}[sep=small]
            & w \in \set{S,T}^* \\
            w' & & w'' \\
            & \cdot
            \ar[from=1-2, to=2-1, two heads, "\alpha"']
            \ar[from=1-2, to=2-3, two heads, "\beta"]
            \ar[from=2-1, to=3-2, two heads, dotted, "0" near end]
            \ar[from=2-3, to=3-2, two heads, dotted, "0" {swap, near end}]
            \ar[from=1-2, to=3-2, draw=none, "{{(CR \circlearrowleft)}}" description]
        \end{tikzcd}
        \qquad
        \begin{tikzcd}[sep=small]
            & \zeta \in w\variables \\
            \alpha(\zeta) \in w'\variables & & \beta(\zeta) \in w''\variables \\
            & \cdot
            \ar[from=1-2, to=2-1, mapsto, "\alpha"']
            \ar[from=1-2, to=2-3, mapsto, "\beta"]
            \ar[from=2-1, to=3-2, equal]
            \ar[from=2-3, to=3-2, equal]
            \ar[from=1-2, to=3-2, draw=none, "{{(CR \circlearrowleft)}}" description]
        \end{tikzcd}
    \]

    We claim that $u \theoryeq{U^\lambda} \sep(u)$. Indeed, in every application of $\lambda$, $\mu^S$ or $\mu^T$, any $\bbU^\lambda$-term acting as representative of the input is $\bbU^\lambda$-equal to any term representing the output. For applications of $\lambda$, this is by definition of $E_\lambda$. The multiplication $\mu_\bbS$ maps an input $\eqclass{s[\eqclass{s'_x}\bbS]}\bbS$ to $\eqclass{s[s'_x]}\bbS$. Obviously $s[s'_x] \theoryeq{U^\lambda} s[s'_x]$. The congruence rule of equational logic then equates terms resulting from other choices of representatives for $\eqclass{s'_x}\bbS$, while the substitution rule equates terms resulting from other choices of representatives for $\eqclass{s}\bbS$. The same goes for $\mu_\bbT$.   
    \qedhere
\end{proof}

\begin{remark}\label{rem:AC}
    In general, we need the Axiom of Choice to obtain $\sep$.
    However, if all terms in $\bbS$ and $\bbT$ have unique normal forms, then $\sep$ is constructive. Instead of choosing a representative, we would choose the unique normal form obtained by normalising in both $\bbS$ and $\bbT$. 
\end{remark}

\begin{lemma}\label{lem:sep_on_S_terms}
For all $\bbS$-terms $s$, $\sep(s)\theoryeq{S} s$, for all $\bbT$-terms $t$, $\sep(t)\theoryeq{T} t$, and for any separated term $t[s_x/x]$, $\sep(t[s_x/x])$ is equal to $t[s_x/x]$ \modST. 
\end{lemma}
\begin{proof}
        For an $\bbS$-term $s$, we have $\type(s) = S\variables$ and $\inject(s, S\variables) = \eqS{s}$.
    Since $S\variables$ is already a normal form in $\calRsep$, by definition $\sep(s) = \rho_{S\variables}(\eqS{s})$ is just the choice function choosing a representative of $\eqS{s}$.
    We therefore have $\sep(s) \theoryeq{S} s$.
    The arguments for $\bbT$-terms and for separated terms $t[s_x/x]$ are similar.
\end{proof}

We now apply \cref{lem:dl_=>_composite_theory_lem1} to show that any two separated terms that are equal in $\bbU^\lambda$, are equal \modST.

\begin{lemma}\label{lem:modulo_S_T}
    Any two separated terms equal in $\bbU^\lambda$ are equal \modST.
\end{lemma}

\begin{proof}
    Suppose two separated terms $t_0[s_x/x]$ and $t'_0[s'_y/y]$ are equal in $\bbU^\lambda$.
    Let $\mathfrak{T}$ be a $\bbU^\lambda$-derivation tree of this equality $t_0[s_{x}/x] \theoryeq{U^\lambda} t'_0[s_{y}/y]$ in equational logic.
    By an induction on the structure of $\mathfrak{T}$, we prove that for each equation $u = u'$ in $\mathfrak{T}$, $\sep(u)$ and $\sep(u')$ are equal \modST. By \cref{lem:sep_on_S_terms} and transitiviy of equality \modST, we then conclude that $t_0[s_x/x]$ and $t'_0[s'_y/y]$ are equal \modST.

    The base cases are the Axiom and Reflexivity rules.
    The induction steps are the Symmetry, Transitivity, Congruence, and Substitution rules. We show only the cases of Congruence and Substitution here, as these are the only interesting cases. The full proof is in the appendix.
    \begin{itemize}[itemsep=5pt, topsep=3pt]
        \item
        Congruence: Given $\op^{(n)} \in \Sigma_{\bbU^\lambda}$, consider
            \AxiomC{$u_1 = u'_1$ }
            \AxiomC{$\ldots$}
            \AxiomC{$u_n = u'_n$}
            \TrinaryInfC{$\op(u_1 , \ldots , u_n) = \op( u'_1, \ldots , u'_n)$}
            \DisplayProof
        \proofSkipAmount
        Let $t_i[s_i] \defeq \sep(u_i)$ and $t'_i[s'_i] \defeq \sep(u'_i)$ for $i=1,\ldots,n$.
        The IH is that $\eqT{t_i[\eqS{s_i}]} = \eqT{t'_i[\eqS{s'_i}]}$.
        We distinguish when $\op$ is a $\bbT$-symbol or an $\bbS$-symbol.
        \begin{itemize}[topsep=4pt]
            \item 
            If $\op \in \Sigma_\bbT$, then
            \(
                \eqT{\op(\eqT{t_1[\eqS{s_1}]}, \ldots, \eqT{t_n[\eqS{s_n}]})} = \eqT{\op(\eqT{t'_1[\eqS{s'_1}]}, \ldots, \eqT{t'_n[\eqS{s'_n}]})},
            \)
            by congruence in $\bbT$.
            By \CRc, applying $\mu^T S$ on both sides results in the $TS$-equivalence classes of $\sep(\op(u_1, \ldots, u_n))$ and $\sep(\op(u'_1, \ldots, u'_n))$, which must be equal.
            Thus $\sep(\op(u_1, \ldots, u_n))$ and $\sep(\op(u'_1, \ldots, u'_n))$ are equal \modST.

            \item 
            If $\op \in \Sigma_\bbS$, then 
            \(
                \eqS{\op(\eqT{t_1[\eqS{s_1}]}, \ldots, \eqT{t_n[\eqS{s_n}]})} = \eqS{\op(\eqT{t'_1[\eqS{s'_1}]}, \ldots, \eqT{t'_n[\eqS{s'_n}]})},
            \)
            by congruence in $\bbS$.
            {By \CRc}, applying $T \mu^S \cdot \lambda$ on both sides results in the $TS$-equivalence classes of $\sep(\op(u_1, \ldots, u_n))$ and $\sep(\op(u'_1, \ldots, u'_n))$, which must be equal.
            Thus $\sep(\op(u_1, \ldots, u_n))$ and $\sep(\op(u'_1, \ldots, u'_n))$ are equal \modST.
        \end{itemize}
        
        \item
        Substitution:
        Consider a substitution $f$, consider
            \AxiomC{$u = u'$}
            \UnaryInfC{$u[f] = u'[f]$}
            \DisplayProof. \\
        Let $t[s_x] \defeq \sep(u)$ and $t'[s'_x] \defeq \sep(u')$.
        The IH is that $\eqT{t[\eqS{s_x}]} = \eqT{t'[\eqS{s'_x}]}$.
        We start by separating all terms in the image of $f$.
        This gives another substitution $g \defeq \sep \cdot f$, and we denote $t_y[s_z] \defeq g(y)$ for all $y \in \var(u_1) \cup \var(u_2)$.
        By closure of $\theoryeq{S}$ and $\theoryeq{T}$ under substitution, we have 
        \[
            \eqT{t[\eqS{s_x[\eqT{t_y[\eqS{s_z}]}]}]} =   \eqT{t'[\eqS{s'_x[\eqT{t_y[\eqS{s_z}]}]}]}.
        \]
        Applying $\mu^T \mu^S \cdot T \lambda S$ on both sides make us reach $TS\variables$, and the two sides of the equation are the $TS$-equivalence classes of $\sep(u[f])$ and of $\sep(u'[f])$, which must therefore be equal by \CRc.
        Hence $\sep(u[f])$ and $\sep(u'[f])$ are equal \modST. \qedhere
    \end{itemize}
\end{proof}

The proof of \cref{thm:dl_=>_composite_theory} now follows from \cref{lem:dl_=>_composite_theory_lem1,lem:modulo_S_T}.

The next theorem was given in Zwart's thesis \cite[Theorem 3.9]{Zwart_2020} but not published elsewhere.
We have updated the reasoning and obtained a much shorter proof using the shortcut $\EM{TS} \cisom \catalg{\lambda}$.

\begin{theorem}
    \label{thm:composite_theory_from_distributive_law}
    Let $S$ and $T$ be the free algebra monads of algebraic theories $\bbS$ and $\bbT$. If there is a distributive law $\lambda: ST \Rightarrow TS$, then
    the monad $\langle TS, \eta^T\eta^S, \mu^T\mu^S \cdot T\lambda S\rangle$ is presented algebraically by $\bbU^\lambda$. \customqed
\end{theorem}

\section{Axiomatisations of Composite Theories}
\label{sec:application}

In \cref{thm:dl_=>_composite_theory}, we showed how to obtain an algebraic presentation $\bbU^\lambda$ of the composite monad arising from a distributive law $\lambda : ST \to TS$.
However, the set of equations $E_\lambda$ accounting for the interactions between $\bbS$- and $\bbT$-terms is maximal in the sense that it contains all possible equations that consist of representatives
of some pair $(u,\lambda(u))$ in the graph of $\lambda$. In practice, we would like to have a minimal description of $E_\lambda$, such as the one for $\Ring$ in \cref{ex:monoid_abelian_ring}, which only adds two distribution axioms to the theories of monoids and Abelian groups.

In this section, we identify criteria on the shape of axioms that allow us to prove that certain minimal subsets of $E_\lambda$ suffice to generate the whole of
$E_\lambda$.
We apply term rewriting methods for proving the necessary claims.

The shape of axioms will be described in terms of \emph{layers}.

\begin{definition}
    Let $\bbS$ and $\bbT$ be two algebraic theories.
    Given a term $s[t_x/x] \in \Sigma_\bbS^* \Sigma_\bbT^* \variables$, its \myemph{$ST$-layers} is the pair $(m,n)$ of natural numbers where $m \defeq \depth(s)$ and $n \defeq \max \setvbar{\depth(t_x)}{x \in \var(s)}$, where $\depth$ denotes
    the maximal number of nested (possibly nullary) operation symbols. This corresponds to the inductively defined notion of depth of term trees where we define constants to have depth 1, and variables depth 0.
    $TS$-layers are defined similarly for terms in $\Sigma_\bbT^* \Sigma_\bbS^* \variables$.
\end{definition}

\begin{example}
    We illustrate $ST$-layers with terms in $\Ring$ (where $\bbS=\Monoid,\bbT=\AbGrp$). 
    \begin{center}\begin{tabular}{c|ccccccc}
        $ST$-Layers
        & \, $(0,0)$ \, 
        & \, $(0,1)$ \, 
        & \, $(1,0)$ \, 
        & \, $(1,1)$ \, 
        & \, $(0,2)$ \,
        & \, $(2,0)$ \,
        \\ \hline
        \multirow{2}{*}{Examples}
        & $x$ 
        & $0$ 
        & $1$
        & $x \cdot 0$ 
        & $x+0$
        & $x \cdot 1$
        \\
        & $y$
        & $x + y$
        & $x \cdot y$
        & \, $(x+y) \cdot (y+z)$ \,
        & \, $(x + y) + z$ \,
        & \, $x \cdot (y \cdot z)$ \, 
    \end{tabular}\end{center}
\end{example}

For the remainder of this section, we assume that $\bbS$, $\bbT$, $\lambda$, $E_\bbS$, $E_\bbT$, $E_\lambda$, and $\bbU^\lambda$ are as in \cref{thm:dl_=>_composite_theory}.

\begin{lemma}
    \label{lem:simplification_E_lambda}
    Let $E' \subseteq E_\lambda$  be a set of equations defined as follows.
    For each $f^{(n)} \in \bbS$, $g^{(m)} \in \bbT$ and each $i \in \set{1,\ldots,n}$,
    $E'$ contains one equation of the form $l = r$, where
    $l = f(x_1,\ldots,x_{i-1},g(\vec{y}),x_{i+1},\ldots,x_n)$
    and
    $r \in \lambda_\variables(\smash{\eqS{\eqT{l}}})$.
    The pair $(\Sigma_{\bbU^\lambda} = \Sigma_\bbS \uplus \Sigma_\bbT, E')$ can now be seen a TRS.
    If this TRS is terminating, then
    $E_\bbS \cup E_\bbT \cup E'$ generates the same congruence on $\bbU^\lambda$-terms as $E_\bbS \cup E_\bbT \cup E_\lambda$.
\end{lemma}

\begin{proof}
    Let us show why $(\Sigma_{\bbU^\lambda} = \Sigma_\bbS \uplus \Sigma_\bbT, E')$ is a TRS.
    First, no left-hand side is a variable by definition of $E'$.
    Second, $A \defeq \var(s[t_x]) \supseteq \var(t[s_y])$ holds for all $(s[t_x],t[s_y]) \in E'$.
    This is the case since $\lambda_A : STA \to TSA : \smash{\eqclass{s[\eqclass{t_x}\bbT]}\bbS \mapsto \eqclass{t[\eqclass{s_y}\bbS]}\bbT}$ forces the equivalence class of $t[s_y]$ to be in $TSA$ and therefore to only use the variables in $A$.
    
    Now let us argue why the congruence relation is left unchanged.
    Take an equation $(u,u') \in E_\lambda \cup E_\bbS \cup E_T$.
    The goal is to obtain this equation using only $E_\bbS \cup E_\bbT \cup E'$.
    \begin{itemize}
        \item First, using only equations in $E'$, the $\bbU^\lambda$-terms $u$ and $u'$ can be separated.
        Indeed, we assume that the TRS $(\Sigma_{\bbU^\lambda}, E')$ is terminating, thus both $u$ and $u'$ can be rewritten to normal forms.
        The equations $E'$ are exhaustive in the following sense: every term containing a $\Sigma_\bbT$-symbol below an $\Sigma_\bbS$-symbol is reducible (not in normal form).
        Thus the normal forms of $u$ and $u'$ must be in $\Sigma^*_\bbT \Sigma^*_S\variables$, i.e., separated.    
        Let us denote them $t[s_x/x]$ and $t'[s'_y/y]$. 
    
        \item 
        Since $\bbU^\lambda$ is a composite theory (proven in \cref{thm:dl_=>_composite_theory}), and the separated normal forms $t[s_x/x]$ and $t'[s'_y/y]$ are $\bbU^\lambda$-equal, they must also be equal \modST. 
        By equality \modST, we have a proof of $t[s_x/x] = t'[s'_y/y]$ using only equations from $E_\bbS$ and $E_\bbT$
        (explicitly so when using the equivalent formulation $(4)$ of equality \modST in \cite[Prop.~3.4]{Zwart_2020}).
        \qedhere
    \end{itemize}
\end{proof}

Choosing a right-hand side $r$ in \cref{lem:simplification_E_lambda} uses the axiom of choice, similarly to the definition of $\sep$ in \cref{sec:dl_=>_composite_theory}.
As pointed out in \cref{rem:AC}, if the theories $\bbS$ and $\bbT$ can be oriented to obtain a confluent and terminating TRS, then the normal form is a natural choice for $r$.
For example, in \cite{Pirog_Staton_2017}, the theory of left-zero monoids and the theory with a unary idempotent operation were both oriented, allowing for a practical presentation of the composite theory that the authors called $\mathsf{CUT}$.

\begin{example}
    Let us retrieve the axiomatisation of $\Ring$ as given in \cref{ex:monoid_abelian_ring}, but starting from its corresponding distributive law $\lambda : L \AbGrpMonad \to \AbGrpMonad L$ \cite[§4]{Beck_1969}.
    The set $E$ will only contain equations whose left-hand side is among
    \(
        (x+y)z, \,
        x(y+z), \,
        0 \cdot x, \,
        x \cdot 0, \,
        (-x)y, \text{ and }
        x(-y).
    \)
    For each of those, there are infinitely many choices for the right-hand side.
    For instance $(x \cdot 0, 0)$, $(x \cdot 0, 0+0)$, etc.
    Thankfully, there is an easy choice for the right-hand side $r$, because the theory $\Monoid$ can be oriented,
    \(
        (x y) z \to x(y z), \;\; 1 \cdot x \to x, \text{ and } x \cdot 1 \to x,
    \)
    as can the theory $\AbGrp$ without the symmetry axiom.
    Not taking the symmetry axiom into account simply means that we have to choose one equation between $((x+y)z,\, xz + yz)$ and $((x+y)z,\, yz+xz)$.
    We end up with 
    $6$ equations:
    \begin{align*}
        (x + y)z &= xz + yz,
        & x \cdot 0 &= 0 ,
        & (-x) y &= -(xy), \\
        z (x + y) &= zx + zy,
        & 0 \cdot x &= 0,
        & x (-y) &= -(xy).
    \end{align*}
    Reducing from $6$ to only the $2$ equations of left and right distributivity can be done using automated tools. 
    In our case, we used Prover9~\cite{McCune_2005_Prover9} and obtained the result instantaneously (see \cref{sec:prover9_appendix} in appendix).
\end{example}

Note that if $E' \subseteq E_\lambda$ is not terminating, then the conclusion is not guaranteed to hold.
The example below exhibits a situation where the set $E'$ of equations as defined in \cref{lem:simplification_E_lambda} is not enough to generate all of the $E_\lambda$ equations.

\begin{example}
    \label{ex:counter_example_aab}
    We show that the subset of equations of $E_\lambda$ where all left-hand sides have layers $(1,1)$ is not always sufficient (together with $E_\bbS$ and $E_\bbT$) to generate all $E_\lambda$ equations obtained from a distributive law $\lambda$. 
    This example is an extension of the well-known non-terminating TRS $ab \to bbaa$ \cite[Ex.2.3.9]{Terese_2003}.

    Consider the theories $\bbS$ and $\bbT$, with signatures $\Sigma_\bbS \defeq \set{a^{(1)}}$ and $\Sigma_\bbT \defeq \set{b^{(1)}}$, and equations $E_\bbS \defeq \set{aaa = aa}$ and $E_\bbT \defeq \set{bbb = bb}$.
    We use some string rewriting notations, such as $aax$ or $a^2x$ as shorthand for $a(a(x))$, etc.
    The set of equivalence classes of $\bbS$ is $S\variables = \{ \eqS{a^2x}, \eqS{ax}, \eqS{x} \mid x \in \variables\}$. 
    Similarly,  $T\variables = \{ \eqT{b^2x}, \eqT{bx}, \eqT{x} \mid x \in \variables\}$. 
    We define a mapping
    \[\begin{array}{rcll}
    \lambda\colon ST\variables & \to & TS\variables\\[.2em]
         \eqS{a^n\eqT{b^mx}} &\mapsto & \eqT{b^2\eqS{a^2x}}, & \text{ for }n, m \in \{1,2\}  \\[.2em]
         \eqS{a^n\eqT{x}} &\mapsto & \eqT{\eqS{a^nx}}, & \text{ for }n \in \{1,2\} \\[.2em]
         \eqS{\eqT{b^nx}} &\mapsto& \eqT{b^n\eqS{x}}, & \text{ for }n \in \{1,2\}\\[.2em]
         \eqS{\eqT{x}} &\mapsto& \eqT{\eqS{x}}
    \end{array}
    \]
    We show that $\lambda$ is a distributive law: 
    \begin{itemize}
        \item 
        Unit law \eqref{eqn:distributive_law_unit_axiom_S}: 
        $\lambda_\variables(S\eta^T_{\variables}(\eqS{a^nx})) = \lambda_\variables(\eqS{a^n\eqT{x}}) = \eqT{\eqS{a^nx}} = \eta^T_{S\variables}(\eqS{a^nx})$.

        \item 
        Unit law \eqref{eqn:distributive_law_unit_axiom_T}:  
        $\lambda_\variables(\eta^S_{T\variables}(\eqT{b^nx})) = \lambda_\variables(\eqS{\eqT{b^nx}}) = \eqT{b^n\eqS{x}} = T\eta^S_\variables(\eqT{b^nx})$.

        \item 
        Multiplication law \eqref{eqn:distributive_law_multiplication_axiom_S}: We only show the case for $n,m,k \geq 1$. Other cases can be easily verified in a similar manner.\\
        \begin{tikzcd}[scale cd=.95]
            \eqS{a^n \eqS{a^m \eqT{b^kx}}} \in SST & \eqS{a^n\eqT{b^2\eqS{a^2x}}} \in STS & \eqT{b^2 \eqS{a^2\eqS{a^2x}}} \in TSS \\
            \eqS{a^{n+m} \eqT{b^kx} }\in ST & & \eqT{b^2\eqS{a^2x}} = \eqT{b^2\eqS{a^4x}} \in TS
            \ar[from=1-1, to=1-2, "S\lambda"]
            \ar[from=1-2, to=1-3, "\lambda S"]
            \ar[from=1-1, to=2-1, "\mu^S T"']
            \ar[from=1-3, to=2-3, "T \mu^S"]
            \ar[from=2-1, to=2-3, "\lambda"']
        \end{tikzcd}

        \item 
        Multiplication law \eqref{eqn:distributive_law_multiplication_axiom_T}:
        Analogous to the previous point.
    \end{itemize}
    From \cref{thm:dl_=>_composite_theory}, defining the set $E_\lambda$ of distributivity equations as below ensures that $E_\bbS \cup E_\bbT \cup E_\lambda$ is an axiomatization of the composite theory $\bbU^\lambda$.
    \[\begin{array}{rcl}
        E_\lambda &=& \{ a^nb^mx = b^2a^2x \mid m,n \geq 1, x \in \variables\} \cup\\
        && \{ a^nx = a^nx, b^nx=b^nx \mid n \in \{0,1,2\}, x \in \variables\}    
        \end{array}
    \]
    The subset of equations of $E_\lambda$ that have left-hand side with $ST$-layers $(1,1)$ is
    $E' = \{ab = b^2a^2\}$. However, we claim that $E_\bbS \cup E_\bbT \cup E'$ cannot derive all equations in $E_\lambda$. 
    Indeed, we observe that the distributivity equation $aab =_{E_\lambda} bbaa$ cannot be derived.
    Trying to do so leaves us stuck in a loop:
    (we underline the part where an equation is applied)
    \begin{align*}
        a\underline{ab} & =_{E'} \underline{ab}baa  =_{E'} bba\underline{ab}aa
        =_{E'}  bbabb\underline{aaaa}
        =_{E_\bbS} bb\underline{ab}baa \\ & =_{E'} \underline{bbbb}aabaa  
         =_{E_\bbT} bbaabaa =_{E'} \ldots \text{ (loop) } 
    \end{align*}
    It is not hard to see that there are no other ways of proving $aab = bbaa$ in $E_\bbS \cup E_\bbT \cup E'$. Hence $E_\bbS \cup E_\bbT \cup E'$ does not generate the same congruence as $E_\bbS \cup E_\bbT \cup E_\lambda$.
    In line with \cref{lem:simplification_E_lambda}, the above indeed also shows that $E'$, when viewed as a TRS, is not terminating. 
    Note that \cref{lem:simplification_E_lambda} only says that termination is a sufficient condition for a (1,1)-axiomatisation.
    It does not exclude that in some composite theories, the set of equations $E_\bbS \cup E_\bbT \cup E'$ might axiomatise $\bbU^\lambda$ even in presence of non-termination.
\end{example}

The next lemma identifies a class of equations where termination of the TRS $(\Sigma_\bbU = \Sigma_\bbS \uplus \Sigma_\bbT, E')$ is guaranteed.
These are equations in which the right-hand sides have layers $(n,1)$, which is inspired from similar results for string rewriting obtained by Zantema \& Geser~\cite{Zantema_Geser_2000_0p1q_1r0s}.

\begin{lemma}
    \label{lem:(1.1)->(n.1)_layers}
    Let $\bbS$ and $\bbT$ be two algebraic theories.
    Let $R$ be a set rules of the form $s[t_x/x] \to t[s_y/y]$.
    Let $Z = \{t_x \mid t_x \text{ is a variable} \}$, i.e., all $z \in Z$ occur directly below an $\bbS$-operation in $s[t_x/x]$.
    If each $s[t_x/x]$ has $ST$-layers $(1,1)$, 
    each $t[s_y/y]$ has $TS$-layers $(n,1)$ for some $n$ not fixed,
    and each $s_y$ is linear\footnote{
        \emph{Linear} in a TRS sense, i.e. variables appearing \textit{at most} once.
    } 
    in $Z$,
    then $R$ is terminating. \customqed
\end{lemma}

\begin{example}\label{ex:axiomatisations}
    We give some axiomatisations of composite theories resulting from distributive laws in the literature:
    \begin{enumerate}
        \item Let $R(X) = X^A$ be the reader monad, with $A = \set{a_1, \ldots, a_n}$.
        There is a distributive law of the finite distribution monad $\distribution$ over $R$, $\lambda : \distribution R \to R\distribution$, that sends $p_1 h_1 + \ldots + p_n h_n$ to $(a \mapsto p_1 h_1(a) + \ldots + p_n h_n(a))$ \cite[Example 1.34]{Goy_2021_Thesis}.
        Recall that $R$ is presented algebraically by a single operation $f^{(n)}$ with two equations~\cite{Plotkin_Power_2001}, and $\distribution$ is presented by convex algebras. 
        The distribution axioms as described in \cref{lem:simplification_E_lambda} are in our case, for each $p \in [0,1]$
        \begin{align*}
            f(x_1, \ldots, x_n) \oplus_p y &= f(x_1 \oplus_p y, \ldots, x_n \oplus_p y). \\
            x \oplus_p f(y_1, \ldots, y_n) &= f(x \oplus_p y_1, \ldots, x \oplus_p y_n).
        \end{align*}
        We see that the right-hand sides of these equations have layers $(1,1)$ and both equations satisfy the linearity requirement of \cref{lem:(1.1)->(n.1)_layers}, thus ensuring termination. Hence by \cref{thm:dl_=>_composite_theory,lem:simplification_E_lambda}, the above equations together with the equations for $f$ and for convex algebras present the composite monad on $R \distribution$ induced by $\lambda$.
        Furthermore, we notice that each of the above equations can be derived from the other one using the axioms of convex algebras.
        Therefore, we only need to include one of them for each $p$.
        
        \item 
        There is a distributive law of multisets over distributions $\lambda\colon \multiset \distribution \to \distribution \multiset$ called the \emph{parallel multinomial law} in \cite{Jacobs_2021_MD_DM_distributive_law}, see also \cite{Dahlqvist_Parlant_Silva, Dash_Staton} and \cite[Ex.~1.37]{Goy_2021_Thesis}.
        It sends $\Lbag px_1 + (1-p)x_2, y \Rbag$ to $p\Lbag x_1,y \Rbag + (1-p) \Lbag x_2,y \Rbag$, which can be expressed in the syntax of convex algebras and commutative monoids as       
        \[
            (x_1 \oplus_p x_2) \cdot y = (x_1 \cdot y) \oplus_p (x_2 \cdot y).
        \]
        By \cref{thm:dl_=>_composite_theory}, \cref{lem:simplification_E_lambda} and \cref{lem:(1.1)->(n.1)_layers} these equations (one for each $p \in [0,1]$), together with the axioms of convex algebras and commutative monoids, present the composite monad on $\distribution\multiset$ induced by $\lambda$. 

        \item There is a distributive law $\lambda\colon L^+ L^+ \to L^+L^+$ for the non-empty list monad over itself \cite{ManesMulry_2007}. It sends a list of lists to the singleton list containing the list of all heads: $[[a,b],[c],[d,e,f]] \mapsto [[a,c,d]]$. We get the following distributivity axioms for the composite theory:
        \vspace{-0.25\baselineskip}
        \begin{align*}
            a * (b \star c) & = a * b \\
            (a \star b) * c & = a * c.
        \end{align*}
        Again, the equations satisfy the conditions for \cref{lem:(1.1)->(n.1)_layers}, and our results imply that the above equations together with the semigroup axioms for $*$ and $\star$ present the composite monad on $L^+ L^+$ induced by $\lambda$.
    \end{enumerate}
\end{example}

\section{Conclusion}\label{sec:conclusion}

In this paper, we proved the correspondence between composite theories of $\bbT$ after $\bbS$ and distributive laws $\lambda\colon ST \to TS$.
Furthermore, we gave sufficient criteria for when a minimal set $E' \subseteq E_\lambda$ of distribution equations, along with $E_\bbS$ and $E_\bbT$, axiomatises the composite theory.

The set $E'$ itself is unlikely to turn many heads, as distributive laws are often informally described in the literature in terms of such simple distribution axioms. The surprise, however, comes from the fact that $E'$ is not always enough (see \cref{ex:counter_example_aab}). This is a possible pitfall similar to the `simplicity' of the various false distributive laws of the powerset monad over itself \cite{Klin_Salamanca_2018}.

There are several directions for future work.
We showed that termination of $E'$ (as TRS) is sufficient for $E_\bbS \cup E_\bbT \cup E'$ to axiomatise the composite theory (\cref{lem:simplification_E_lambda}), and that taking equations in $E'$ to have layers $(1,1) \to (n,1)$ ensures termination (\cref{lem:(1.1)->(n.1)_layers}).
We would like to identify other criteria for termination, and make more use of term rewriting techniques.
In fact, we did consider equations with layers $(1,1) \to (1,n)$, but we were only able to prove termination (using dependency pairs) by adding a strong linearity requirement which seems too restrictive to be useful (see \cref{lem:(1.1)->(1.n)_layers} in the appendix).
We speculate that one could allow layers $(1,1) \to (2,2)$ in which some symbol in the left-hand side is absent from the right-hand side in order to avoid problems such as in \cref{ex:counter_example_aab} with $ab \to bbaa$.

In light of negative results concerning monad compositions \cite{Klin_Salamanca_2018, Varacca_Winskel_2006, Zwart_Marsden_journal_2022, Dahlqvist_Neves_2018_no_monad_on_PD}, there has been much interest in understanding the limits of monad composition. Positive results using algebraic methods were given in  \cite{Dahlqvist_Parlant_Silva}.
Another approach has been to generalise to so-called weak distributive laws \cite{Garner_2020, Goy_2021_Thesis}.
Presentations of monads arising from the composition of monads via a weak distributive law, in particular monads for nondeterminism and probabilities, have been given in \cite{Bonchi_Sokolova_Vignudelli_2019, Goy_Petrisan_2020}. 
These presentations are obtained by adding a simple distribution axiom to the two underlying theories, similar to our results in Section~\ref{sec:application}, but the resulting theory is no longer a composite theory as the essential uniqueness \modST is not guaranteed to hold.
Another future line of work would be to extend the current correspondence to weak distributive laws \cite{Garner_2020, Goy_2021_Thesis} thereby giving a definition of \emph{weak composite theories}. Such a correspondence would allow for a more thorough study of weak distributive laws on the algebraic level, and could perhaps lead to no-go theorems for weak distributive laws.

\subsubsection{Acknowledgments:} 
Alo\"is Rosset and J\"org Endrullis received funding from the Netherlands Organization for Scientific Research (NWO) under the Innovational Research Incentives Scheme Vidi (project. No. VI.Vidi.192.004).

\newpage

{
\bibliographystyle{meta/splncs04}
\sloppy
\bibliography{main}
}
\newpage

\section{Appendix}

\subsection{Preliminaries monads extra}

\begin{table}[H]
    \centering
    \begin{tabular}{c c}
        \toprule
        & \\
        {
            \AxiomC{$(s,t) \in E$}
            \RightLabel{\scriptsize Axiom$_{E}$}
            \UnaryInfC{$s=t$}
            \DisplayProof
        }
        &{
            \AxiomC{}
            \RightLabel{\scriptsize Reflexivity}
            \UnaryInfC{$t=t$}
            \DisplayProof
        } \\ & \\
        {
            \AxiomC{$s=t$}
            \RightLabel{\scriptsize Symmetry}
            \UnaryInfC{$t=s$}
            \DisplayProof
        }
        &{
            \AxiomC{$t_1=t_2$}
            \AxiomC{$t_2=t_3$}
            \RightLabel{\scriptsize Transitivity}
            \BinaryInfC{$t_1=t_3$}
            \DisplayProof
        } \\ & \\
        {
            \AxiomC{$s_1=t_1$}
            \AxiomC{$\ldots$}
            \AxiomC{$s_n=t_n$}
            \RightLabel{\scriptsize Congruence}
            \TrinaryInfC{$\op(s_1,\ldots,s_n) = \op(t_1,\ldots,t_n)$}
            \DisplayProof
        }
        &{
            \AxiomC{$s=t$}
            \RightLabel{\scriptsize Substitution}
            \UnaryInfC{$s[f]=t[f]$}
            \DisplayProof
        } \\ & \\
        \bottomrule
    \end{tabular}
    \small
    \caption{Inference rules of equational logic, with $n \in \N$, $\op^{(n)} \in \Sigma$, and $f$ a substitution} 
    \label{tab:equational_logic}
\end{table}
\vspace*{-\baselineskip}

\subsection{Preliminaries TRS extra}


We present the version of the multiset path order \cite{SchneiderKamp_etal_2007_Recursive_path_orders} that we will use.

\begin{definition}[Multiset ordering, \cite{Dershowitz_Manna_1979_Multiset_ordering, Dershowitz_1982_Orderings_for_TRS}]
    \label{def:multiset_ordering}
    A transitive, irreflexive relation $(A,\succ)$ can be extended to the multiset $(\multiset A, \succ_\multiset)$, where a multiset is reduced by removing one or more elements and replacing them with any finite number (possibly $0$) of elements, all smaller than one of the elements removed.
\end{definition}



The multiset ordering $(\multiset A, \succ_\multiset)$ is well-founded iff $(A,\succ)$ is well-founded \cite[Theorem §2]{Dershowitz_Manna_1979_Multiset_ordering}.

\begin{example}
    Taking the usual \qmarks{greater than} on natural numbers $(\N,>)$, then $\Lbag 2, 2, 0 \Rbag >_\multiset \Lbag 2, 1, 1 \Rbag$ in the multiset ordering, since an occurrence of $2$ has been replaced by two $1$'s, and an occurrence of $0$ has been removed. 
\end{example}

\begin{definition}[{\cite[§2]{Hofbauer_1992_Termination_mpo}}]
    Let $\bbU$ be an algebraic theory.
    We define \myemph{permutation equivalence} of $\bbU$-terms by $x \sim x$ for variables $x \in \variables$, and $f[t_1, \ldots, t_n] \sim g(s_1, \ldots, s_m)$ if and only if $f=g$ (which implies $m=n)$ and for all $i$, $t_i \sim s_{\pi(i)}$ for some permutation $\pi$ of $\set{1, \ldots, n}$.
\end{definition}

\begin{definition}[Multiset path order, \cite{SchneiderKamp_etal_2007_Recursive_path_orders}]
    \label{def:multiset_path_order}
    Let $\bbU$ be an algebraic theory with a reflexive, transitive relation on its function symbols $(\Sigma_\bbU, \leq)$.
    Let $< \defeq (\leq \setminus \geq)$ denote the strict version and $\approx \defeq (\leq \cap \geq)$ denote the induced equivalence.
    We inductively define a transitive relation $\succ$ on $\bbU$-terms as follows.
    Let $l,r$ be $\bbU$-terms.
    We say that $l \succ r$ if and only if $l=f(l_1, \ldots, l_n)$ for some $f^{(n)} \in \Sigma_\bbU$ and one of the following hold:
    \begin{enumerate}[1.]
        \item 
        \label{it:mpo1}
        $l_i \succ r$ or $l_i \sim r$ for some $i$; or
        
        \item
        \label{it:mpo2}
        $r = g(r_1, \ldots, r_m)$ for \textit{some} $g^{(m)} \in \Sigma_\bbU$ and either
        \begin{enumerate}[i.]
            \item 
            \label{it:mpo2i}
            $f > g$ and $l \succ r_j$ for all $j$, or
            
            \item 
            \label{it:mpo2ii}
            $f \approx g$ and $\Lbag l_1, \ldots, l_n \Rbag \succ_\multiset \Lbag r_1, \ldots, r_m \Rbag$.
        \end{enumerate}
    \end{enumerate}
\end{definition}

\begin{lemma}[{\cite[§3]{Hofbauer_1992_Termination_mpo}}]
    The relation $\succ$ of \cref{def:multiset_path_order} is a reduction order, i.e, is well-founded and closed under substitutions and contexts.
    Therefore, a TRS $(\Sigma, R)$ is terminating if $R \subseteq {\succ}$.
\end{lemma}

Let us explain the basics about \emph{critical pairs}~\cite[§2.7]{Terese_2003}.
We only give an informal description here, as the precise definition is quite technical.
An overlap is when two rules, for instance $f(g(x),y) \to f(x,g(y))$ and $g(x) \to h(x)$, can be applied with some symbols being shared, in our case $g$. 
They create the critical pair $f(x, g(y)) \leftarrow f(g(x),y) \to f(h(x),y)$.
We say that a critical pair \emph{converges} if the two terms can be rewritten to a common term, i.e., $f(x, g(y)) \rightarrowdbl \cdot \leftarrowdbl f(h(x),y)$ in our case.

\subsection{Composite theories extra}

To prove \cref{lem:same_equivalence_class_TSX_implies_a_to_d}, let us recall the other definitions of composite theory from \cite[Def.~3.2, Prop.~3.4]{Zwart_2020}.
\begin{enumerate}[(I)]
    \item 
    \label{it:eqmodST_a_d}
    The separated terms $t[s_x]$ and $t'[s'_y]$ are \myemph{equal \modST} if there exist functions $h : \var(t) \to Z \leftarrow \var(t') : h'$ such that 
    \begin{enumerate}[(a), topsep=4pt, itemsep=1pt] 
        \item \label{item:(a)} \label{ax:weakessuniq1}
        $t[h] \theoryeq{T} t'[h']$.
        
        \item \label{item:(b)} \label{ax:weakessuniq2}
        $\forall x_1,x_2 \in \var(t): h(x_1) = h(x_2) \Rightarrow s_{x_1} \theoryeq{S} s_{x_2}$.
        
        \item \label{item:(c)} \label{ax:weakessuniq3}
        $\forall y_1,y_2 \in \var(t): h'(y_1) = h'(y_2) \Rightarrow s_{y_1} \theoryeq{S} s_{y_2}$.
        
        \item \label{item:(d)} \label{ax:weakessuniq4}
        $\forall x \in \var(t), \forall y \in \var(t'): h(x) = h'(y) \Rightarrow s_{x} \theoryeq{S} s_{y}$.
    \end{enumerate}

    \item 
    \label{it:eqmodST_i_iii}
    The separated terms $t[s_x]$ and $t'[s'_y]$ are \myemph{equal \modST} if there exist functions $h: X \to Z$, $h' : Y \to Z$ and $\bbS$-terms $\setvbar{s^*_z}{z \in Z}$ such that
    \begin{enumerate}[(i)]
        \item $t[h] \theoryeq{T} t'[h']$,
        \item $\forall x \in X: s_x \theoryeq{S} s^*_{h(x)}$, and
        \item $\forall y \in Y: s'_y \theoryeq{S} s^*_{h'(y)}$.
    \end{enumerate}
\end{enumerate}

\begin{proof}[Proof of \cref{prop:essuniqformulations}]
    \begin{enumerate}[align=left]    
        \item[$\ref{it:dan-es}\Rightarrow\ref{it:eqclass-es}$:]
        We suppose equality \modST as given above in \cref{it:eqmodST_i_iii} and show that
        \(
            t[\eqclass{s_x}{\bbS}/x] \theoryeq{T} t'[\eqclass{s'_x}{\bbS}/x].
        \)
        We reason
        \begin{align*}
             t[\eqS{s_x}/x] 
             &\theoryeq{T} t[\eqS{s^*_{h(x)}} /x] \tag*{by $(ii)$}\\
             &\theoryeq{T} t[h][\eqS{s^*_z} /z] \\
             &\theoryeq{T} t'[h'][\eqS{s^*_z} /z] \tag*{by $(i)$} \\
             &\theoryeq{T} t'[\eqS{s^*_{h'(y)}} /y] \\
             &\theoryeq{T} t'[\eqS{s'_y} /y]. \tag*{by $(iii)$}
        \end{align*}
        
        \item[$\ref{it:eqclass-es}\Rightarrow\ref{it:dan-es}$:]
        We suppose that $t[\eqclass{s_x}{\bbS}/x] \theoryeq{T} t'[\eqclass{s'_x}{\bbS}/x]$ and prove \cref{it:eqmodST_a_d} above.
        Taking $Z \defeq S\variables$ and $h: x \mapsto \eqclass{s_x}{\bbS}$, and $h': x \mapsto \eqclass{s'_x}{\bbS},$ \ref{item:(a)} holds immediately.
        We check \ref{item:(b)}: if $h(x_1) = h(x_2)$, it means $\smash{\eqclass{s_{x_1}}{\bbS} = \eqclass{s_{x_2}}{\bbS}}$, i.e., $s_{x_1} \theoryeq{S} s_{x_2}$.
        The reasoning for \ref{item:(c)} and \ref{item:(d)} is the same.
        \qedhere
    \end{enumerate}
\end{proof}

\subsection{Proofs of \cref{sec:composite_theory_=>_dl}}

\begin{lemma}
    \label{lem:composite_=>_dl_lem1}
    $(TS, \eta^{TS} \defeq \eta^T\eta^S, \mu^{TS})$ is a monad.
\end{lemma}

\begin{proof}[Proof of \cref{lem:composite_=>_dl_lem1}]
    Note that by definition of $\phi$, the unit $\eta^{TS} \defeq \eta^T\eta^S$ is the same as $\phi \cdot \eta^U$, as both send a variable $v$ to $\eqclass{\eqclass{v}\bbS}\bbT$.
    
    We prove the first of the three monad axioms. 
    \begin{align*}
        \mu^{TS} \cdot TS \eta^{TS}
        &= (\phi \cdot \mu^U \cdot \psi\psi) \cdot (TS\phi \cdot TS \eta^U)
        \tag*{def.~\(\mu^{TS}\) and \(\eta^{TS}\)} \\
        &= \phi \cdot \mu^U \cdot \psi U \cdot TS\psi \cdot TS\phi \cdot TS \eta^U
        \tag*{horizontal composition} \\
        &= \phi \cdot \mu^U \cdot \psi U \cdot TS \eta^U
        \tag*{$\phi, \psi$ inverses} \\
        &= \phi \cdot \mu^U \cdot U\eta^U \cdot \psi
        \tag*{$\phi$ naturality} \\
        &= \phi \cdot \psi
        \tag*{unit axiom \eqref{eq:def_monad_axiom_unit} for $U$} \\
        &= \id{}
        \tag*{\text{\(\phi, \psi\)} inverses}
    \end{align*}
    The second axiom is proven similarly to the first one:
    \begin{align*}
        \mu^{TS} \cdot \eta^{TS}TS
        &= (\phi \cdot \mu^U \cdot \psi\psi) \cdot (\phi TS \cdot \eta^U TS)
        \tag*{def.~$\mu^{TS}$ and $\eta^{TS}$} \\
        &= \phi \cdot \mu^U \cdot U\psi \cdot \psi TS \cdot \phi TS \cdot \eta^U TS
        \tag*{vertical composition} \\
        &= \phi \cdot \mu^U \cdot U\psi \cdot \eta^U TS
        \tag*{$\phi, \psi$ inverses} \\
        &= \phi \cdot \mu^U \cdot \eta^U U \cdot \psi
        \tag*{$\eta^U$ naturality} \\
        &= \phi \cdot \psi 
        \tag*{unit axiom \eqref{eq:def_monad_axiom_unit} for $U$} \\
        &= \id{}
        \tag*{$\phi, \psi$ inverses}
    \end{align*}
    The third axiom is also just finding back the same axiom for $U$:
    \begin{align*}
        \mu^{TS} TS \cdot \mu^{TS}
        &= (\phi \cdot \mu^U \cdot \psi\psi) \cdot (TS \phi \cdot TS \mu^U \cdot TS \psi\psi)
        \tag*{def.~$\mu^{TS}$} \\
        &= \phi \cdot \mu^U \cdot \psi U \cdot TS\psi \cdot TS \phi \cdot TS \mu^U \cdot TS \psi\psi
        \tag*{vert.~composition} \\
        &= \phi \cdot \mu^U \cdot \psi U \cdot TS \mu^U \cdot TS \psi\psi
        \tag*{$\phi, \psi$ inverses} \\
        &= \phi \cdot \mu^U \cdot U \mu^U \cdot \psi UU \cdot TS \psi\psi
        \tag*{$\psi$ naturality} \\
        &= \phi \cdot \mu^U \cdot \mu^U U \cdot \psi UU \cdot TS \psi\psi
        \tag*{mult.~axiom \eqref{eq:def_monad_morphism_axiom_multiplication} for $U$} \\
        &= \phi \cdot \mu^U \cdot \mu^U U \cdot \psi\psi\psi
        \tag*{vertical composition} \\
        &= \phi \cdot \mu^U \cdot \mu^U U \cdot UU\psi \cdot \psi\psi TS
        \tag*{vertical composition} \\
        &= \phi \cdot \mu^U \cdot U \psi \cdot \mu^U TS \cdot \psi\psi TS
        \tag*{$\mu^U$ naturality} \\
        &= (\phi \cdot \mu^U \cdot U \psi \cdot \psi TS) \cdot (\phi TS \cdot \mu^U TS \cdot \psi\psi TS)
        \tag*{$\phi, \psi$ inverses} \\
        &= \mu^{TS} \cdot \mu^{TS} TS
        \tag*{def.~$\mu^{TS}$}
    \end{align*}
\end{proof}

\begin{lemma}
    \label{lem:composite_=>_dl_lem2}
    The nat.~transformations $\eta^T S$ and $T\eta^S$ are monad morphisms.
    \customqed
\end{lemma}

\begin{proof}[Proof of \cref{lem:composite_=>_dl_lem2}] \belowdisplayskip=-12pt 
    We show $\eta^T S: S \Rightarrow TS$ is a monad morphism.
    The proof for $T \eta^S$ is then analogous.
    There are two axioms to prove.
    The first one \eqref{eq:def_monad_morphism_axiom_unit}, $\eta^T S \cdot \eta^T = \eta^{TS}$ is immediate, as both sides simply send a variable $v \in \variables$ to $\eqclass{\eqclass{v}{\bbS}}{\bbT}$.
    The second one \eqref{eq:def_monad_morphism_axiom_multiplication} is $\mu^{TS} \cdot (\eta^T S)(\eta^T S) = \eta^T S \cdot \mu^S$.
    Given any $\eqclass{s[\eqclass{s_x}{\bbS}/x]}{\bbS}$ in $SS \variables$:
    \begin{align*} 
        \big( \mu^{TS} \cdot (\eta^T S)(\eta^T S) \big) \big( \eqclass{s[\eqclass{s_x}{\bbS}/x]}{\bbS} \big) 
        &= \mu^{TS} \big( \eqclass{\eqclass{s[\eqclass{\eqclass{s_x}{\bbS}}{\bbT}/x]}{\bbS}}{\bbT} \big) 
        \tag*{$\eta^T S$ twice}\\
        &= \phi \cdot \mu^U \cdot \psi\psi \big( \eqclass{\eqclass{s[\eqclass{\eqclass{s_x}{\bbS}}{\bbT}/x]}{\bbS}}{\bbT} \big) 
        \tag*{def.~$\mu^{TS}$ \eqref{eq:def_muTS}} \\
        &= \phi \cdot \mu^U \big( \eqclass{s[\eqclass{s_x}{\bb{U}}/x]}{\bb{U}} \big) 
        \tag*{def.~$\psi$ \eqref{eq:def_psi}} \\
        &= \phi \big( \eqclass{s[s_x/x]}{\bb{U}} \big) 
        \tag*{applying multiplication} \\
        &= \eqclass{\eqclass{s[s_x/x]}{\bbS}}{\bbT} 
        \tag*{def.~$\phi$ \eqref{eq:def_phi}} \\
        &= \eta^T S \big( \eqclass{s[s_x/x]}{\bbS} \big)  \\
        &= \big( \eta^T S \cdot \mu^S \big) \big( \eqclass{s[\eqclass{s_x}{\bbS}]}{\bbS}\big)
    \end{align*} \qedhere
\end{proof}

\begin{lemma}
    \label{lem:composite_=>_dl_lem3}
    The middle unitary law holds: $\mu^{TS} \cdot T\eta^S \eta^T S = Id_{TS}$.
    \customqed
\end{lemma}

\begin{proof}[Proof of \cref{lem:composite_=>_dl_lem3}] 
    Let $\eqclass{t[\eqclass{s_x}{\bbS}/x]}{\bbT}$ be any element of $TS\variables$. Then:
    \begin{align*}
        \big(\mu^{TS} \cdot T\eta^S \eta^T S \big) \big(\eqclass{t[\eqclass{s_x}{\bbS}/x]}{\bbT}\big) 
        &= \mu^{TS} \big(\eqclass{t[\eqclass{\eqclass{\eqclass{s_x}{\bbS}}{\bbT}}{\bbS}/x]}{\bbT}\big) 
        \tag*{applying units} \\
        & = \big(\phi \cdot \mu^U \cdot \psi\psi\big) \big(\eqclass{t[\eqclass{\eqclass{\eqclass{s_x}{\bbS}}{\bbT}}{\bbS}/x]}{\bbT}\big) 
        \tag*{def.~$\mu^{TS}$ \eqref{eq:def_muTS}} \\
        &= \big(\phi \cdot \mu^U\big) \big(\eqclass{t[\eqclass{s_x}{\bb{U}}/x]}{\bb{U}}\big)
        \tag*{def.~$\psi$ \eqref{eq:def_psi}} \\
        &= \phi \big(\eqclass{t[s_x/x]}{\bb{U}}\big)
        \tag*{applying multiplication} \\
        &= \eqclass{t[\eqclass{s_x}{\bbS}/x]}{\bbT}
        \tag*{def.~$\phi$ \eqref{eq:def_phi}}
    \end{align*}
\end{proof}

\begin{proof}[End of proof of \cref{thm:composite_=>_dl}]
    We show that $\lambda = \mu^{TS} \cdot \eta^T ST \eta^S$ give the expression in the statement of the theorem.
    Let $\eqclass{s[\eqclass{t_x}{\bbT}/x]}{\bbS}$ be any element of $ST\variables$, and $t'[s'_x/x]$ be a separated term such that $t'[s'_x/x] \theoryeq{U} s[t_x/x]$ then:
    \begin{align*}
        \lambda \big(\eqclass{s[\eqclass{t_x}{\bbT}/x]}{\bbS}\big) 
        &=  \big(\mu^{TS} \cdot \eta^T ST \eta^S\big) \big(\eqclass{s[\eqclass{t_x}{\bbT}/x]}{\bbS}\big) 
        \tag*{def.~$\lambda$} \\
        &= \mu^{TS} \big(\eqclass{\eqclass{s[\eqclass{\eqclass{t_x}{\bbS}}{\bbT}/x]}{\bbS}}{\bbT}\big) 
        \tag*{applying units} \\
        &= \big(\phi \cdot \mu^U \cdot \psi\psi\big) \big(\eqclass{\eqclass{s[\eqclass{\eqclass{t_x}{\bbS}}{\bbT}/x]}{\bbS}}{\bbT}\big) 
        \tag*{def.~$\mu^{TS}$ \eqref{eq:def_muTS}} \\
        &= \big(\phi \cdot \mu^U\big) \big(\eqclass{s[\eqclass{t_x}{\bbU}/x]}{\bbU}\big) 
        \tag*{def.~$\psi$ \eqref{eq:def_psi}} \\
        &= \phi \big(\eqclass{s[t_x/x]}{\bbU}\big) 
        \tag*{applying multiplication} \\
        &= \eqclass{t'[\eqclass{s'_x}{\bbS}/x]}{\bbT}
        \tag*{def.~$\phi$ \eqref{eq:def_phi}}
    \end{align*}
\end{proof}

\subsection{Proofs of \cref{sec:dl_=>_composite_theory}}

\noindent\begin{minipage}{.69\linewidth}
\begin{proof}[Proof of \cref{lem:functorial_newmans_lemma}]
    The proof is identical to $(ii)$ in \cite[Theorem 1.2.1]{Terese_2003}.
    We prove by well-founded induction that for all term-functor $A$, we have ${\text{\CRc}(A)}$ by showing that if all its descendent have \CRc, then $A$ must have it too.
    The I.H.~is that all proper descendants satisfy \CRc.
    Take $C \leftarrowdbl A \rightarrowdbl B$.
    The cases where one side is a 0-step are immediate.
    Otherwise, the reasoning can be summarized in one diagram, confluence is obtained with common reduct $D_3$, and the diagram commutes:
\end{proof}
\end{minipage}
\hfill
\begin{minipage}{.3\linewidth}
    \[\begin{tikzcd}[row sep=small, column sep=tiny, scale cd=.8]
    	&& A \\
    	& {B_1} && {C_1} \\
    	B && {D_1} && C \\
    	& {D_2} \\
    	&& {D_3}
    	\arrow[from=1-3, to=2-2]
    	\arrow[two heads, from=2-2, to=3-1]
    	\arrow[two heads, from=3-3, to=4-2]
    	\arrow[two heads, from=2-4, to=3-3]
    	\arrow[two heads, from=3-5, to=5-3]
    	\arrow[two heads, from=3-1, to=4-2]
    	\arrow[two heads, from=2-2, to=3-3]
    	\arrow[from=1-3, to=2-4]
    	\arrow[two heads, from=2-4, to=3-5]
    	\arrow[two heads, from=4-2, to=5-3]
    	\arrow[shift left=1, bend left=10, two heads, from=2-4, to=4-2]
    	\arrow[phantom, "\text{\WCRc}", from=1-3, to=3-3]
    	\arrow[phantom, "\text{IH} \circlearrowleft", from=2-2, to=4-2]
    	\arrow[phantom, "\text{IH} \circlearrowleft", from=2-4, to=5-3]
    \end{tikzcd}\]
\end{minipage}

\begin{proof}[Proof of \cref{lem:functorial_critical_pair_lemma}]
    $(\Rightarrow)$ 
    Suppose \WCRc.
    A critical pair is a one-step peak.
    By \WCRc, both reducts must have a common descendent, with the diagram commuting.
    
    $(\Leftarrow)$
    Take a peak $T_0 \xleftarrow{S_0} T \xrightarrow {S_1} T_1$ with $S_0 = L_0 F$ with rule $\rho_0 = (L_0 \to R_0)$ and substitution $F$, and $S_1 = L_1 G$ with rule $\rho_1 = (L_1 \to R_1)$ and substitution $G$.
    The notation is usually with lowercase letters and lowercase Greek letters $l_0 \sigma$ and $l_1 \tau$, but we choose for clarity to write functors as uppercase and natural transformations as lowercase.
    Indeed, if we write the variable at the end of our terms that is usually omitted in string rewriting we have $l_0 (\sigma(x))$ and the substitution $\sigma(x)$ gives a term of rewriting system, i.e., a functor and thus we prefer to denote it as an uppercase letter, here $F$.
    
    By \cite[Def.~2.1.5 and Lem.~2.7.6]{Terese_2003}, the redex occurrences $S_0$ and $S_1$ can be either \textit{disjoint} (i.e.,~parallel positions), \textit{nested} or \textit{overlapped}.
    With only unary operations, the disjoint case is not possible in string rewriting systems.
    \begin{itemize}
        \item Suppose $S_0$ and $S_1$ are nested.
        Suppose w.l.o.g.~that $S_1$ is a subterm of $S_0$.
        Let $C[~]$ be the context of $S_0$ in $T$, i.e., $T = C[S_0]=C[L_0 F]$.
        We adopt the categorical convention of not writing parentheses or brackets when not needed, thus writing $T= C L_0 F$.
        Since there is no overlap, the redex $S_1$ occurs inside the term $F$.
        In other words, $F = E S_1 = E L_1 G$ for some context $E$.
        Observe that $C R_0 E R_1 G$ is a common reduct and that the diagram commutes by naturality of $\rho_0$.
        \[\begin{tikzcd}[column sep=large]
        	{T = C S_0 = C L_0 E L_1 G} & {T_1 = C L_0 E R_1 G} \\
        	{T_0 = C R_0 E L_1 G} & {C R_0 E R_1 G}
        	\arrow["{C \rho_0 E L_1 G}"', from=1-1, to=2-1]
        	\arrow["{C L_0 E \rho_1 G}", from=1-1, to=1-2]
        	\arrow["{C \rho_0 E R_1 G}", dotted, from=1-2, to=2-2]
        	\arrow["{C R_0 E \rho_1 G}"', dotted,  from=2-1, to=2-2]
        	\arrow[phantom, "{C (\rho_0 \text{ nat.})}", from=1-1, to=2-2]
        \end{tikzcd}\]
        
        \item
        Suppose $S_0$ and $S_1$ overlap.
        Denote again $C$ the context of $S_0$ in $T$ and $S'_0$ and $S'_1$ their respective contracta.
        If the overlap is trivial $S_0 = S_1$ and $\rho_0 = \rho_1$, then
        \[\begin{tikzcd}[sep = small]
        	{C S_0 = C S_1} & {C S_1'} \\
        	{C S'_0} & {C S'_0}
        	\arrow["{C \rho_0}"', from=1-1, to=2-1]
        	\arrow["{C \rho_1}", from=1-1, to=1-2]
        	\arrow[equal, from=1-2, to=2-2]
        	\arrow[equal,  from=2-1, to=2-2]
        	\arrow[phantom, "{\circlearrowleft}", from=1-1, to=2-2]
        \end{tikzcd}\]
        Suppose the overlap non-trivial, and w.l.o.g.~$S_1 \leq S_0$.
        By \cite[Lemma 2.7.12]{Terese_2003}, the contracta $S'_0, S'_1$ can be found as instances $\langle H_0 P, H_1 P\rangle$ of a critical pair $\langle H_0, H_1 \rangle$.
        By hypothesis, the critical pair converge with a commuting diagram.
        Thus by applying $C \circ - \circ \rho$ around, we have what we desire.
        \[
            \begin{tikzcd}
            	{\cdot} & {H_1} \\
            	{H_0} & {H_2}
            	\arrow["{\alpha}"', from=1-1, to=2-1]
            	\arrow["{\beta}", from=1-1, to=1-2]
            	\arrow["{\exists \gamma}", dotted, from=1-2, to=2-2]
            	\arrow["{\exists \delta}"', dotted, from=2-1, to=2-2]
            	\arrow[phantom, "{\circlearrowleft}", from=1-1, to=2-2]
            \end{tikzcd}
            \quad \Rightarrow \quad
            \begin{tikzcd}
            	{\cdot} & {T_1 = C H_1 P} \\
            	{T_0 = C H_0 P } & {C H_2 P}
            	\arrow["{C \alpha P}"', from=1-1, to=2-1]
            	\arrow["{C \beta P}", from=1-1, to=1-2]
            	\arrow["{C \gamma P}", dotted, from=1-2, to=2-2]
            	\arrow["{C \delta P}"', dotted, from=2-1, to=2-2]
            	\arrow[phantom, "{\circlearrowleft}", from=1-1, to=2-2]
            \end{tikzcd}
        \]
    \end{itemize}
\end{proof}

\begin{proof}[Proof of the missing cases of \cref{lem:modulo_S_T}]
    \begin{itemize}[itemsep=5pt, topsep=3pt]
        \item
        Axiom:
        We distinguish when the axiom is taken in $E_\bbS$, $E_\bbT$ or $E_\lambda$:
        \begin{itemize}
            \item 
            For $(s_1,s_2) \in E_\bbS$, by \cref{lem:sep_on_S_terms}, we have $\sep(s_1) \theoryeq{S} s_1 \theoryeq{S} s_2 \theoryeq{S} \sep(s_2)$, and hence  $\sep(s_1) \theoryeq{S,T} \sep(s_2)$.
            \item 
            For $(t_1,t_2) \in E_\bbT$, there is no $\bbS$-part and, similarly to the previous point, we obtain $\sep(t_1) \theoryeq{S,T} \sep(t_2)$.

            \item 
            For $(s[t_x],t[s_y]) \in E_\lambda$,
            then $\lambda(\eqS{s[\eqT{t_x}]}) = \eqT{t[\eqS{s_y}]}$.
            We already have a path $\lambda$ from $\type(s[t_x]) = ST\variables$ to the normal form $TS \variables$.
            Therefore $\eqT{t[\eqS{s_y}]}$, and the $TS$-equivalence classes of $\sep(s[t_x])$ and $\sep(t[s_y])$ must be equal by well-definedness of $\sep$.
            In particular, $\sep(s[t_x])$ and $\sep(t[s_y])$ are equal \modST.
        \end{itemize}
        

        
        \item
        Reflexivity:
        \AxiomC{}
        \UnaryInfC{$u = u$}
        \DisplayProof
        Immediate, since $\sep(u)$ is equal \modST to itself.
        
        \item
        Symmetry:
        \AxiomC{$u_1 = u_2$}
        \UnaryInfC{$u_2 = u_1$}
        \DisplayProof
        The IH is that $\sep(u_1)$ and $\sep(u_2)$ are equal \modST, which is what we desire.
        
        \item
        Transitivity:
        \AxiomC{$u_1 = u_2$}
        \AxiomC{$u_2 = u_3$}
        \BinaryInfC{$u_1 = u_3$}
        \DisplayProof
        Immediate from the IH, since equality \modST is transitive (visible using \cref{lem:same_equivalence_class_TSX_implies_a_to_d}).
\end{itemize}
\end{proof}

\begin{proof}[Proof of \cref{thm:composite_theory_from_distributive_law}]
    Recall that the category of Eilenberg-Moore $TS$-algebras is concretely isomorphic to the category of $\lambda$-algebras, i.e., $\EM{TS} \cisom \catalg{\lambda}$. 
    To finish the proof, it suffices to show that $\catalg{\lambda} \cisom \catalg{\bbU^\lambda}$.

    Our hypothesis that $S$ and $T$ are free algebra monads mean:
    \begin{equation}
        \label{eq:S_T_algebraic_presentations}
        \begin{array}{rcl}
             \EM{S} &\cisom& \catalg{\bbS} \\
             (\sigma: SX \to X) &\mapsto& (X, \brackets{\cdot}^\sigma) \\
             (\sigma_{\brackets{\cdot}}: SX \to X) &\mapsfrom& (X, \brackets{\cdot})
        \end{array}
        \qquad
        \begin{array}{rcl}
             \EM{T} &\cisom& \catalg{\bbT} \\
             (\tau: TX \to X) &\mapsto& (X, \brackets{\cdot}^\tau) \\
             (\tau_{\brackets{\cdot}}: TX \to X) &\mapsfrom& (X, \brackets{\cdot})
        \end{array}
    \end{equation}
    where $\brackets{s}^\sigma \defeq \sigma(\eqS{s})$, \ 
    $\sigma_{\brackets{\cdot}}(\eqS{s}) \defeq \brackets{s}$, \ 
    $\brackets{t}^\tau \defeq \tau(\eqT{t})$, and \ 
    $\tau_{\brackets{\cdot}}(\eqT{t}) \defeq \brackets{t}$.

    We construct mappings $F : \catalg{\lambda} \rightleftarrows \catalg{\bbU^\lambda} : G$.
    Let $F$ and $G$ be the identity on morphisms.
    On objects let $G(X, \brackets{\cdot}) \defeq (X, \sigma_{\brackets{\cdot}}, \tau_{\brackets{\cdot}})$, where $\sigma_{\brackets{\cdot}}$ and $\tau_{\brackets{\cdot}}$ are defined as in \eqref{eq:S_T_algebraic_presentations}, and let $F(X, \sigma, \tau) \defeq (X, \brackets{\cdot}^{\sigma\tau})$, where
    $\brackets{\cdot}^{\sigma\tau} \defeq \brackets{\cdot}^\sigma \cup \brackets{\cdot}^\tau$.

    Let us show that $F$ is a functor.
    We suppose that $(X, \sigma, \tau)$ is a $\lambda$-algebra, i.e., the pentagon from \cref{def:lambda_algebra} commutes.
    We show that $(X, \brackets{\cdot}^{\sigma\tau})$ is a $\bbU^\lambda$-algebra.
    The interpretation $\brackets{\cdot}^{\sigma\tau}$ satisfies the equations in $E_\bbS$ and $E_\bbT$ because respectively $\brackets{\cdot}^{\sigma}$ and $\brackets{\cdot}^{\tau}$ do so.
    Moreover, any equation $(s[t_x],t[s_y])$ in $E_\lambda$ is also satisfied:
    \begin{center}
        \begin{tikzpicture}[scale=.9]
            \node (ST) at (-4,2) {
                $\eqclass{s[\eqclass{t_x}\bbT]}\bbS \in ST\variables$
                };
            \node (S) at (-4, 1) {
                $\eqclass{s[\tau(\eqclass{t_x}\bbT)]}\bbS 
                \stackrel{\eqref{eq:S_T_algebraic_presentations}}{=}
                \eqclass{s[ \brackets{t_x}^\tau]}\bbS \in S\variables$
            };
            \node (TS) at (4,2) {
                $\eqclass{t[\eqclass{s_y}\bbS]}\bbT \in TS\variables$
            };
            \node (T) at (4,1) {
                $\eqclass{t[\sigma(\eqclass{s_y}\bbS)]}\bbT 
                \stackrel{\eqref{eq:S_T_algebraic_presentations}}{=} 
                \eqclass{t[ \brackets{s_y}^\sigma]}\bbT \in T\variables$
            };
            \node (v) at (0,0) {
                $\brackets{s[ \brackets{t_x}^\tau]}^\sigma = \brackets{s[t_x]}^{\sigma\tau} = \brackets{t[s_y]}^{\sigma\tau} = \brackets{t[\brackets{s_y}^\sigma]}^\tau$
            };
            \draw[->] (ST) -- (TS) node[midway, above] {$\lambda$};
            \draw[->] (ST) -- (S) node[midway, right] {$S\tau$};
            \draw[->] (TS) -- (T) node[midway, left] {$T\sigma$};
            \draw[->] (S) -- (v) node[pos=.8, above right] {$\sigma$};
            \draw[->] (T) -- (v) node[pos=.8, above left] {$\tau$};
        \end{tikzpicture}
    \end{center}

    Let us show that $G$ is a functor.
    We suppose that $(X, \brackets{\cdot})$ is a $\bbU^\lambda$-algebra, i.e., $\brackets{\cdot}$ satisfies all equations in $E_\bbS \cup E_\bbT \cup E_\lambda$.
    We show that 
    $(X, \sigma_{\brackets{\cdot}}, \tau_{\brackets{\cdot}})$ is a $\lambda$-algebra.
    Take some $\smash{\eqclass{s[\eqclass{t_x}\bbT]}\bbS \in ST\variables}$ and let $\smash{\eqclass{t[\eqclass{s_y}\bbS]}\bbT \in TS\variables}$ be its image through $\lambda$.
    Thus $(s[t_x],t[s_y])$ is an equation in $E_\lambda$ which must be satisfied by $\brackets{\cdot}$, hence the pentagon necessarily commutes on input $\smash{\eqclass{s[\eqclass{t_x}\bbT]}\bbS \in ST\variables}$:
    \begin{center}
        \begin{tikzpicture}[scale=.9]
            \node (ST) at (-4,2) {
                $\eqclass{s[\eqclass{t_x}\bbT]}\bbS \in ST\variables$
                };
            \node (S) at (-4, 1) {
                $\eqclass{s[\tau_{\brackets{\cdot}}(\eqclass{t_x}\bbT)]}\bbS 
                \stackrel{\eqref{eq:S_T_algebraic_presentations}}{=}
                \eqclass{s[ \brackets{t_x}]}\bbS \in S\variables$
            };
            \node (TS) at (4,2) {
                $\eqclass{t[\eqclass{s_y}\bbS]}\bbT \in TS\variables$
            };
            \node (T) at (4,1) {
                $\eqclass{t[\sigma_{\brackets{\cdot}}(\eqclass{s_y}\bbS)]}\bbT 
                \stackrel{\eqref{eq:S_T_algebraic_presentations}}{=} 
                \eqclass{t[ \brackets{s_y}]}\bbT \in T\variables$
            };
            \node (v) at (0,0) {
                $\brackets{s[ \brackets{t_x}]} = \brackets{s[t_x]} = \brackets{t[s_y]} = \brackets{t[\brackets{s_y}]}$
            };
            \draw[->] (ST) -- (TS) node[midway, above] {$\lambda$};
            \draw[->] (ST) -- (S) node[midway, right] {$S\tau_{\brackets{\cdot}}$};
            \draw[->] (TS) -- (T) node[midway, left] {$T\sigma_{\brackets{\cdot}}$};
            \draw[->] (S) -- (v) node[pos=.85, above right] {$\sigma_{\brackets{\cdot}}$};
            \draw[->] (T) -- (v) node[pos=.85, above left] {$\tau_{\brackets{\cdot}}$};
        \end{tikzpicture}
    \end{center}
    To finish the proof, simply notice that $F$ and $G$ are inverses because they use the inverses constructions from \eqref{eq:S_T_algebraic_presentations}.
\end{proof}

\subsection{Extra section 6}

\begin{proof}[Proof of \cref{lem:(1.1)->(n.1)_layers}]
    We will denote symbols in $\Sigma_\bbS$ by $f,f',\ldots$ and symbols in $\Sigma_\bbT$ by $g,g',\ldots$.
    We will establish termination using the multiset path order (see \cite{Hofbauer_1992_Termination_mpo} or \cref{def:multiset_path_order}) where
    all $f,f' \in \Sigma_\bbS$ are equivalent $f \approx f'$, all $g,g' \in \Sigma_\bbT$ are equivalent $g \approx g'$, and all $\bbS$-symbols are greater than $\bbT$-symbols $f > g$.
    The goal is to show that $R \subseteq {\succ}$.
    
    Consider $s[t_x/x] \to t[s_y/y] \in R$. 
    We use induction on the structure of $t$:
    \begin{enumerate}[(i)]
        \item 
        For $t[s_y/y]$ being a variable, we obtain $s[t_x/x] \succ t[s_y/y]$ by repeated application of case \ref{it:mpo1} of \cref{def:multiset_path_order}.
        \item 
        For $t$ being a variable $y$ and $s_y = f'(z_1,\ldots,z_k)$, then we are in case \ref{it:mpo2}\ref{it:mpo2ii} of \cref{def:multiset_path_order}:
        \begin{itemize}
            
            \item 
            We have indeed $f \approx f'$.
            
            \item 
            For $s = f(t_1,\ldots,t_n)$, we also have $\Lbag t_1, \ldots, t_m \Rbag \succ_\multiset \Lbag z_1, \ldots, z_k \Rbag$.
            Consider $z_i$ for $1 \le i \le k$.
            Since $z_i$ is a variable, it must appear in some $t_j$ for $1 \le j \le m$.
            If $t_j$ is a variable, then $z_i = t_j$ and $z_i$ appears only once in $\Lbag z_1, \ldots, z_k \Rbag$ because $s_z$ is linear in $t_j$.
            If $t_j$ is not a variable, we have $t_j \succ z_i$ by repeated application of case \ref{it:mpo1} of \cref{def:multiset_path_order}.
            Finally, at least one $t_j$ for $1 \le j \le m$ is not a variable (and thus replaced by $0$ or more smaller elements).
        \end{itemize}
        \item 
        For $t = g(t_1,\ldots,t_m)$, we are in case \ref{it:mpo2}\ref{it:mpo2i} of \cref{def:multiset_path_order}.
        We have $s[t_x/x] \succ t_i[s_y/y]$ for every $1 \leq i \leq m$, by IH, and $f > g$ by hypothesis. \qedhere
    \end{enumerate}
\end{proof}

\begin{lemma}
    \label{lem:(1.1)->(1.n)_layers}    
    Let $\bbS$ and $\bbT$ be two algebraic theories.
    Let $R$ be a set of rules of the form $s[t_x/x] \to t[s_y/y]$.
    Suppose each $s[t_x/x]$ has layers $(1,1)$ and each $t[s_y/y]$ has layers $(1,n)$ for some $n$ not fixed, and is linear.
    Then $R$ is terminating.
\end{lemma}

\begin{proof}
    We show termination using dependency pairs \cite{Arts_Giesl_2000_dependency_pairs,Giesl_2004_Dependency_pairs}.
    We therefore extend our signature by adding a marked version $h_\#$ of each symbol $h \in \Sigma_\bbS \uplus \Sigma_\bbT$.    
    For terms $t = h(\vec{u})$ we write $t_\#$ for the term $h_\#(\vec{u})$ obtained by marking the root symbol of $t$.
    We will denote symbols in $\Sigma_\bbS$ by $f,f',\ldots$ and symbols in $\Sigma_\bbT$ by $g,g',\ldots$. 
    
    The dependency pairs here are
    \begin{align*}
        \mathsf{DP}(R) = 
        \{\,
            f_\#(\vec{t}) \to s'_\#
        \mid\;
            &f(\vec{t}) \to g(s_1,\ldots,s_q) \in R,\\
            &s' \text{ is a non-variable subterm of } s_i \text{ for some } i = 1,\ldots,q
        \,\}.
    \end{align*}
    We use the following polynomial interpretation with
    \begin{align*}
        \brackets{f_\#}(x_1,\ldots,x_n) &= \brackets{f}(x_1,\ldots,x_n) = x_1 + \ldots + x_n \\
        \brackets{g_\#}(x_1,\ldots,x_n) &= \brackets{g}(x_1,\ldots,x_n) = x_1 + \ldots + x_n + 1 
    \end{align*} 
    for all $f \in \Sigma_\bbS$ and $g \in \Sigma_\bbT$.
    This interpretation simply counts symbols in $\Sigma_\bbT$.
    Then $R \subseteq {\geq}$ since the rules $s[t_x/x] \to t[s_y/y]$ contain at least $1$ $\Sigma_\bbT$-symbol in $s[t_x/x]$ (not all $t_x$ are variables), only one $\Sigma_\bbT$-symbol in $t[s_y/y]$, and the rules are non-duplicating. Similarly $\mathsf{DP}(R) \subseteq {>}$ since their right-hand sides contain no $\Sigma_\bbT$-symbol.
\end{proof}

\subsection{Identifying derivable equations using Prover 9}
\label{sec:prover9_appendix}

The ring axioms given as follows,
\[
    E_\mathsf{Mon} \defeq \left\{
    \begin{aligned}
        (x \cdot y) \cdot z &= x \cdot (y \cdot z), \\
        1 \cdot x &= x, \\
        x \cdot 1 &= x
    \end{aligned}
    \right\},
    \quad
    E_\mathsf{AbGrp} \defeq \left\{
    \begin{aligned}
        (x + y) + z &= x + (y + z), \\
        x + (-x) &= 0, \\
        x + y &= y + x, \\
        x + 0 &= x
    \end{aligned}
    \right\},
\]
\[
    E_\mathsf{Ring} \defeq E_\mathsf{Mon} \uplus E_\mathsf{AbGrp} \uplus
    \left\{
    \begin{aligned}
        x \cdot (y + z) &= (x \cdot y) + (x \cdot z), \\
        (y + z) \cdot x &= (y \cdot x) + (z \cdot x)
    \end{aligned}
    \right\},
\]
imply the following equalities
\begin{align*}
  x \cdot 0 &= 0 &
  0 \cdot x &= 0 &
  (-x) \cdot y &= -(x \cdot y) &
  x \cdot (-y) &= -(x \cdot y)
\end{align*}
This can be proven automatically by Prover9~\cite{McCune_2005_Prover9} using the following input:
\begin{verbatim}
assign(max_seconds, 30).

formulas(sos).

  (x * y) * z = x * (y * z) # label(times_associativity).
  1 * x = x # label(times_neutral_left).
  x * 1 = x # label(times_neutral_right).

  (x + y) + z = x + (y + z) # label(plus_associativity).
  x + (-x) = 0 # label(plus_inverse).
  x + y = y + x # label(plus_commutativity).
  x + 0 = x # label(plus_neutral).

  (x + y) * z = (x * z) + (y * z) # label(distributivity_right).
  z * (x + y) = (z * x) + (z * y) # label(distributivity_left).
  
end_of_list.

set(restrict_denials).

formulas(goals).

  x * 0 = 0 # answer(times_zero_right).
  0 * x = 0 # answer(times_zero_left).
  (-x) * y = -(x * y) # answer(times_minus_left).
  x * (-y) = -(x * y) # answer(times_minus_right).

end_of_list.
\end{verbatim}
Prover9 only needs a fraction of a second to proof all 4 equations:
\begin{verbatim}
-------- Proof 1 -------- times_zero_left.
-------- Proof 2 -------- times_zero_right.
-------- Proof 3 -------- times_minus_left.
-------- Proof 4 -------- times_minus_right.

THEOREM PROVED
\end{verbatim}


\end{document}